\documentclass{llncs}

\usepackage{algorithmic}

\usepackage[section]{algorithm}
\usepackage{amssymb,amsmath}
\usepackage{array}
\usepackage{graphicx}
\usepackage{psfrag}
\usepackage{slashbox}

\include{cite}

\newcommand{\REM}[1]{}

\newcommand {\A}{\mathcal{A}}
\newcommand {\modG}{\tilde{G}}

\newcommand {\modA}{\tilde{\mathcal{A}}}
\newcommand {\modE}{\tilde{E}}

\newcommand {\p}{\mathcal{P}}
\newcommand{\EE}{\mathcal{E}}
\newcommand{\odd}{\mathcal{O}}
\newcommand{\U}{\mathcal{U}}
\newcommand{\X}{\mathcal{X}}
\newcommand{\Y}{\mathcal{Y}}

\newcommand{\tightpair} {tight-pair }

\newcommand{\oddone}{\mathcal{O}_1}
\newcommand{\evenone}{\mathcal{E}_1}
\newcommand{\unone}{\mathcal{U}_1}

\newcommand{\oddtwo}{\mathcal{O}_2}
\newcommand{\eventwo}{\mathcal{E}_2}
\newcommand{\untwo}{\mathcal{U}_2}

\newcommand{\fonly}{$\A_f$}
\newcommand{\sonly}{$\A_s$ }
\newcommand{\fands}{$\A_{f/s}$ }

\begin{document}
\title{Popular matchings:  structure and cheating strategies \thanks{This work was supported in part by NSF grant CCF-0830737.}}
\author{}

\author{Meghana Nasre}
\institute{University of Texas at Austin,  USA}
\maketitle
\begin{abstract}
We consider the cheating strategies for the popular matchings problem.
Let $G = (\A \cup \p, E)$ be a bipartite graph where $\A$ denotes a set of agents, $\p$ denotes a set of posts
and the edges in $E$ are ranked. Each agent ranks a subset of posts in an order of preference, possibly involving ties.
A matching $M$ is popular if there exists no matching $M'$ such that the number of agents that
prefer $M'$ to $M$ exceeds the number of agents that prefer $M$ to $M'$.
Consider a centralized market where agents submit their preferences and a central authority matches
agents to posts according to the notion of popularity. 
Since a popular matching need not be unique, we assume that the central authority
chooses an arbitrary popular
matching. Let $a_1$ be the sole manipulative agent who is aware of the true
preference lists of all other agents.
The goal of $a_1$ is to falsify her preference list to get 
{\em better always}, that is, in the falsified instance
(i) every popular matching matches $a_1$ to a post that is at least as 
good as the most-preferred post that she gets when she was truthful, and 
(ii) some popular matching matches $a_1$ to a post better than the most-preferred post $p$ that 
she gets when she was truthful, assuming that $p$ is not one of $a_1$'s (true) most-preferred posts.
We show that the optimal cheating strategy for a single agent to get {\em better always} can be computed 
in $O(m+n)$ time when preference lists
are all strict and  
in $O(\sqrt{n}m)$ time
when preference lists are allowed to contain ties.
Here $n = |\A| + |\p|$ and $m = |E|$.
Next, we consider the set of agents, their preference lists and
the popular matchings algorithm as a non-cooperative game. We show a necessary
and sufficient condition for the true preference lists of the agents to
be an equilibrium of this game when each agent wishes to get {\em better always}.

To compute the cheating strategies, we develop a {\em switching graph} characterization of the popular
matchings problem involving ties. The switching graph characterization was studied for the case of strict lists
by McDermid and Irving (J. Comb. Optim. 2011) and  was open for the case of ties. We show an $O(\sqrt{n}m)$ algorithm to compute the
set of {\em popular pairs} using the switching graph.
These results are of independent interest and answer a part of the open questions posed by McDermid and Irving.
\end{abstract}
\section{Introduction}
We consider the cheating strategies for the popular matchings problem.
Let $G = (\A \cup \p, E)$ be a bipartite graph where $\A$ denotes a set of agents, $\p$ denotes a set of posts,
and the edges in $E$ are ranked. Each agent ranks a subset of posts in an order of preference, possibly involving ties.
This ranking of posts by an agent is called the preference list of the agent.
An agent $a$ prefers post $p_i$ to post $p_j$ if the rank of post $p_i$ is
smaller than the rank of post $p_j$ in $a$'s preference list.
An agent $a$ is indifferent between posts $p_i$ and $p_j$ if they have the same rank on $a$'s preference list.
When agents can be indifferent between posts, the preference lists are said to contain ties, otherwise the preference lists
are strict.
A matching $M$ of $G$ is a subset of edges, no two of which share an end point.
For a matched vertex $u$, let $M(u)$ denote its partner in the matching $M$.
An agent $a$ prefers a matching $M$ to another matching $M'$ if (i) $a$ is matched in $M$ but unmatched in $M'$,
or (ii) $a$ prefers $M(a)$ to $M'(a)$.
\begin{definition}
A matching $M$ is {\em more popular than} $M'$ if the number of agents that prefer $M$ is greater than the number of agents that prefer $M'$.
A matching $M$ is {\em popular} if there is no matching $M'$ that is more popular than $M$.
\end{definition}

There exist simple instances that do not admit any popular matching -- however, when an instance admits a popular
matching, there may be more than one popular matching. 
Abraham~et~al.~\cite{AIKM07} characterized the instances that admit
popular matchings and gave efficient algorithms to compute a popular matching if one exists.

\noindent {\bf Our problem.} Consider a centralized matching market
where each agent $a \in \A$ submits a preference over a subset of posts and a central authority
matches agents to posts using the criteria of popularity. Let $a_1$ be the sole manipulative agent  who is aware of the true
preference lists of all other agents and the preference lists of $a \in \A \setminus \{a_1\}$ remain
fixed throughout. The goal of $a_1$ is
clear: she wishes to falsify her preference list so as to improve the post that she gets 
matched to as compared to the post she got when she was truthful.
Since there may be more than one popular matching in an instance, we assume
that the central authority chooses an arbitrary popular matching.
Let $G = (\A \cup \p, E)$ denote the instance where ranks on the edges represent true
preferences of all the agents. Let $H$ denote the instance obtained by falsifying the preference list of
$a_1$ alone. We assume that $G$ admits a popular matching and $a_1$ falsifies in order
to create an instance $H$ which also admits a popular matching.
Note that it may be possible for $a_1$ to
falsify her preference list such that $H$ does not admit any popular matching. But we do not consider
such a falsification.

Agent $a_1$ wishes to falsify her preference list to ensure that
(i) every popular matching in $H$ matches her to a post that is at least as 
good as the most-preferred post that she gets matched to in $G$, and 
(ii) some popular matching in $H$ matches $a_1$ to a post better than the most-preferred post $p$ that 
she gets matched to in $G$, assuming that $p$ is not $a_1$'s true first choice post.
We term this strategy of $a_1$ as `{\em better always}' strategy.
\vspace{-0.1in}
\subsection{Our contributions}
\vspace{-0.1in}
\begin{itemize}
\item Let $a_1$ be the sole manipulative agent  who wishes to get {\em better always}.
The optimal strategy for $a_1$ can be computed in $O(m+n)$ time when preference lists are
all strict and in $O(\sqrt{n}m)$ time when preference lists are allowed to contain ties.
\item Next, consider the set of agents, their preference lists and the popular matchings algorithm as a non-cooperative game.
We show a necessary and sufficient condition for the true preference lists to be an equilibrium of the game assuming
that every agent wishes to get {\em better always}.
\item To compute the cheating strategies, we develop a {\em switching graph} characterization of the popular
matchings problem involving ties. The switching graph characterization was studied for the case of strict lists
by McDermid and Irving~\cite{MI11} and such a characterization was not known for the case of ties. Using the switching graph, 
we show an $O(\sqrt{n}m)$ time algorithm to compute the
set of {\em popular pairs}.
An edge $(a,p) \in E$ is a popular pair if there exists a popular matching $M$  in $G$ such
that $ (a, p) \in M$. 
We also show that counting the total number of popular matchings in an instance with ties
is $\#$P-Complete. 
The switching graph characterization is of independent interest and answers a part of the open questions in \cite{MI11}.
\end{itemize}
\subsection{Related work}
The work in this paper is motivated by the work of Teo~et~al.~\cite{TST01} where they study the strategic
issues of the stable marriage problem~\cite{GS62}. The stable marriage problem is a generalization
of our problem where both the sides of the bipartition (usually referred to as men and women) rank members
of the opposite side in order of their preference.
Teo~et~al.~\cite{TST01} study the strategic issues of the stable marriage problem where women are required to give
complete preference lists and there is a sole manipulative woman. Further, she is aware of the true preference
lists of all the other women. Teo~et~al.~\cite{TST01} compute an optimal cheating strategy for a single woman under this model.
Huang~\cite{Huang07} studies the strategic issues of the stable room-mates problem~\cite{GS62} under
a similar model. 
In the same spirit, we study the strategic issues of the popular matchings problem.

The notion of popular matchings was introduced by G\"{a}rdenfors~\cite{Gar75} in the context of the stable marriage~\cite{GS62}.
Abraham et al.~\cite{AIKM07} studied the problem for one-sided preference lists and gave a characterization
of instances which admit a popular matching. Subsequent to this result, the popular matchings problem has received
a lot of attention \cite{Mah06} \cite{McC08}  \cite{KMN11} \cite{HuangK11} \cite{Kavitha12}.
However, to the best
of our knowledge none of them is motivated by the strategic issues of the popular matchings problem.

\REM{
{\noindent \em Organization of the paper:} In Section~\ref{sec:prelims} we review the background, in Section~\ref{}
we develop the switching graph characterization for the popular matchings problem with ties. In Section~\ref{sec:cheating-prelims}
we give some intuition and prove some lemmas useful for our cheating strategies.
In Section~\ref{sec:single-agent} we formulate the cheating strategies for a single agent. In Section~\ref{}
we consider the non-cooperative game involving the set of agents, their preferences and the popular matchings algorithm.
}
\section{Background}
\label{sec:prelims}
We first review the following well known properties of maximum matchings in bipartite graphs.
Let $G = (\A \cup \p, E)$ be a bipartite graph and let $M$ be a maximum matching in $G$. 
The matching $M$ defines a partition of the vertex set $\A \cup \p$ into three disjoint sets:
a vertex $v \in \A \cup \p$ is \emph {even} (resp. \emph {odd}) if there is an even (resp. odd) length alternating path in
$G$ w.r.t. $M$ from an unmatched vertex to $v$.
A vertex $v$ is \emph {unreachable} if there is no alternating path from an unmatched vertex to $v$.
Denote by $\EE$, $\odd$ and $\U$ the sets of even, odd, and unreachable vertices, respectively, in $G$.
The following lemma is well known in matching theory; its proof can be found in \cite{GGL95new} or \cite{IKMMP04}.

\begin{lemma}[\cite{GGL95new} Dulmage Mendelsohn]
\label{lemma:node-classification}
Let $\EE$, $\odd$ and $\U$ be the sets of vertices defined by a maximum matching $M$ in $G$. Then,
\begin {itemize}
\item [(a)] $\EE$, $\odd$ and $\U$ are pairwise disjoint, and independent of the maximum matching $M$ in $G$.
\item [(b)] In any maximum matching of $G$, every vertex in $\odd$ is matched with a vertex in
$\EE$, and every vertex in $\U$ is matched with another vertex in $\U$.
The size of a maximum matching is $|\odd| + |\U|/2$.
\item [(c)] No maximum matching of $G$ contains an edge between a vertex in $\odd$ and a vertex
in $\odd \cup \U$.
Also, $G$ contains no edge between a vertex in $\EE$ and a vertex in $\EE \cup \U$.
\end {itemize}
\end{lemma}

We now review the characterization of the popular matchings problem from \cite{AIKM07}.
As was done in \cite{AIKM07},
we create a unique last-resort post $\ell(a)$ for each agent $a$.
In this way, we can assume that every agent is matched, since any unmatched agent $a$ can be
paired with $\ell(a)$.
For an agent $a$, let $f(a)$ be the set of rank-1 posts for $a$. To define $s(a)$, let us consider
the graph  $G_1 = (\A \cup \p, E_1)$
on rank-1 edges in $G$ and let $M_1$ be any maximum matching in $G_1$. Let $\odd_1, \EE_1, \U_1$ define the 
partition of vertices $\A \cup \p$ with respect to $M_1$ in $G_1$. For any agent $a$, let $s(a)$ denote the
set of most preferred posts which belong to $\EE_1$ by the above partition.
Abraham~et~al.~\cite{AIKM07} proved the following theorem.
\begin{theorem}[\cite{AIKM07}]
\label{thm:pop-mat}
A matching $M$ is popular in $G$ iff 

(1) $M \cap E_1$ is a maximum matching of $G_1 = (\A\cup\p, E_1)$, and 

(2) for each agent $a$, $M(a) \in \{f(a) \cup s(a)\}$.

\end{theorem}
The algorithm for solving the popular matching problem is as follows:
each $a \in \A$ determines the sets $f(a)$ and $s(a)$. An $\A$-complete
matching (a matching that matches all agents)
that is maximum in $G_1$ and that matches each $a$ to a post  in
$\{f(a) \cup s(a)\}$ needs to be determined. If no such matching exists, then $G$ does not admit a popular
matching. Abraham~et~al.~\cite{AIKM07} gave an $O(\sqrt{n}m)$ time
algorithm to compute a popular matching in $G$ which is presented as  Algorithm~\ref{algo:pop-matching}.
Steps 7--11 are added by
us and will be used to define the switching graph in the next section.
Abraham~et~al.~\cite{AIKM07} also showed a simpler characterization for the popular matchings 
in case of strict lists which results in an $O(m+n)$ time algorithm to return a popular matching
if one exists.

Let $G' = (\A \cup \p, E')$ denote the graph in which every agent $a$ has edges incident to $\{f(a) \cup s(a)\}$.
Step~4 of Algorithm~\ref{algo:pop-matching} deletes edges from $G'$ which cannot be present in any maximum matching
of $G_1$.
We extend this further and in Step~9 delete edges 
from  $G'$ which cannot be present in any popular matching in $G$.
For this, let us 
partition the vertex set $\A \cup \p$  as $\oddtwo, \eventwo$ and $\untwo$ with respect to a popular matching $M$ in $G'$.
Since any popular matching $M$ is a maximum matching in $G'$, it is easy to see that $M$ cannot
contain edges of the form $\oddtwo \oddtwo$ and $\oddtwo \untwo$ (by Lemma~\ref{lemma:node-classification}(c)).
However, note that since
$M$ matches every agent, it implies that $\A \cap \eventwo = \emptyset$ and $\p \cap \oddtwo = \emptyset$.
Thus, there are no $\oddtwo \oddtwo$ edges in the graph
$G'$. Therefore, 
any edge $(a,p)$ deleted in Step~9 is of the form $a \in \oddtwo$
and $p \in \untwo$. We can now make the following claim. 


\begin{claim}
\label{claim:no-del-Step9}
Let $a$ be an agent such that $a \in \untwo$. Then, in Step~9 of Algorithm~\ref{algo:pop-matching},
no edge incident on $a$ gets deleted.
Let $a$ be an agent such that $a \in \evenone$. Then, in Step~4 of Algorithm~\ref{algo:pop-matching},
no edge incident on $a$ gets deleted.
\end{claim}

\REM{
\begin{claim}
\label{claim:no-del-Step4}
Let $a$ be an agent such that $a \in \evenone$. Then, in Step~4 of Algorithm~\ref{algo:pop-matching},
no edge incident on $a$ gets deleted.
\end{claim}
}

\begin{algorithm}
\begin{algorithmic}[1]
\REQUIRE $G = (\A \cup \p, E)$.
\STATE Construct the graph $G' = (\A\cup\p, E')$, where $E' = \{(a,p) : a \in \A$ and $p \in f(a) \cup s(a)\}$.
\STATE Construct the graph $G_1 = (\A \cup \p, E_1)$ and let $M_1$ be any maximum matching in $G_1$.
\STATE Partition $\A \cup \p$ as $\oddone, \evenone, \unone$ with respect to $M_1$ in $G_1$.
\STATE Remove any edge in $G'$ between a node in $\oddone$ and a node in $\oddone \cup \unone$.\\
\STATE Determine a maximum matching $M$ in $G'$ by augmenting $M_1$.
\STATE Return $M$ if it is $\A$-complete, otherwise return \emph {``no popular matching''}.
\IF {$G$ admits a popular matching}
\STATE Partition $\A \cup \p$ as $\oddtwo, \eventwo, \untwo$ with respect to $M$ in $G'$.
\STATE Remove any edge in $G'$ between a node in $\oddtwo$ and a node in $\untwo$.
\STATE Denote the resulting graph as $G'' = (\A \cup \p, E'')$.
\ENDIF
\end{algorithmic}
\caption{$O(\sqrt{n}m)$-time algorithm for the popular matching problem ~\cite{AIKM07} (Steps 1--6).}
\label{algo:pop-matching}
\end{algorithm}

\begin{definition}
\label{def:choices}
For an agent $a$, let $choices(a)$ be the set of posts $p$ such that $(a, p)$ is an edge in $G''$.
\end{definition}

It is easy to see that for any $a \in \A$, $choices(a) \subseteq \{f(a) \cup s(a) \}$. Further,
if $M$ is a popular matching in $G$, then $M(a) \in choices(a)$.

\section{The switching graph characterization}
\label{sec:switching-graph}
In this section we develop the {\em switching graph} for the popular matchings problem with ties.
In case of strict lists,
McDermid and Irving~\cite{MI11} defined a switching graph $G_M = (\p, E_M)$ as a
directed graph on the posts of $G$ and the edge set $E_M$ was determined by a popular matching $M$ in $G$. 
\REM{
They exploited the switching graph to obtain an arbitrary popular matching $M'$ in $G$ from a given popular
matching $M$. They also used the switching graph to develop efficient algorithms for several problems like
computing popular pairs, and counting the total number of popular matching, to name a few. 
}
In fact, a similar graph was defined even before that by Mahdian~\cite{Mah06} (again for
strict lists) to study existence of popular matchings in random instances.
We use the notation and terminology from \cite{MI11}.

Let $G$ be an instance of the popular matchings problem with ties and let $M$ be a popular matching in $G$.
The switching graph $G_M = (\p, E_M)$ is a directed weighted graph on the posts $\p$ of $G$
and is defined with respect to a popular matching $M$ in $G$. 
The edge set $E_M$ is defined using the pruned graph $G'' = (\A \cup \p, E'')$ constructed
in Step~10 of Algorithm~\ref{algo:pop-matching}.
There exists an edge from $p_i$ to $p_j$ (with $p_i \neq p_j)$ iff for some $a \in \A$, $p_i = M(a)$ and $(a, p_j) \in E''$. 
The weight of an edge $w(M(a), p_j)$ is defined as:
\begin{eqnarray*}
w(M(a), p_j) &= 0&   \ \ \ \mbox{if $a$ is indifferent between $M(a)$ and $p_j$} \\
&= -1& \ \ \ \mbox {if $a$ prefers $M(a)$ to $p_j$} \\
&= +1& \ \ \ \mbox {if $a$ prefers $p_j$ to $M(a)$}.
\end{eqnarray*}
It is easy to see that the graph $G_M = (\p, E_M)$ can be constructed in $O(\sqrt{n}m)$ time using Algorithm~\ref{algo:pop-matching}.

Consider a vertex $p$ in $G_M$. Call $p$ a sink vertex in $G_M$ if out-degree of $p$ is zero in $G_M$.
The following lemma characterizes sinks in $G_M$.
\begin{lemma}
A post $p$ is a sink vertex in $G_M$ if and only if $p$ is unmatched in $M$.
\end{lemma}
Let $\X$ be a maximal weakly connected component of $G_M$. Call $\X$ a {\em sink component} if $\X$ contains one or more
sink vertices otherwise call $\X$ a {\em non-sink component}. 

For a path $T$ (resp. cycle $C$) in $G_M$, the weight of the path $w(T)$ (resp. $w(C)$)
is the sum of the weights on the edges in $T$ (resp. $C$). (Whenever we refer to paths and cycles in
$G_M$ we imply directed paths and directed cycles respectively.) A path $T = \langle p_1, \ldots, p_k\rangle$ in $G_M$
is called a {\em switching path}
if $T$ ends in a sink vertex and $w(T) = 0$. 
Similarly, a cycle $C = \langle p_1,  \ldots, p_k, p_1\rangle$ in $G_M$ is called a {\em switching cycle} if $w(C) = 0$.
Let $\A_T =  \{a_i : M(p_i) = a_i, \mbox { for $i = 1 \ldots k$ }\}$
and  denote by $M' = M \cdot T$  the matching obtained by {\em applying} the switching path to
$M$, that is, for $a_i \in \A_T$, $M'(a_i) = p_{i+1}$ whereas for $a \notin \A_T$, $M'(a) = M(a)$.
Similarly, for a switching cycle $C$, 
define $\A_C = \{a_i : M(p_i) = a_i, \mbox  {for $i = 1 \ldots k$ } \}$
and  denote by $M' = M\cdot C$ the matching
obtained by {\em applying} the switching cycle to $M$, that is, for $a_i \in \A_C$, $M'(a_i) = p_{i+1} \mod k$ whereas for $a \notin \A_C$,
$M'(a) = M(a)$. 

\begin{example}
\label{ex:switching-graph}
\end{example}
Consider an instance $G$ where $\A = \{a_1, \ldots, a_7\}$ and $\p = \{p_1, \ldots, p_9\}$. The preference
lists of the agents are shown in Figure~\ref{fig:example-ex2}(a). The preference lists can be read as follows: agent $a_1$
ranks posts $p_1, p_2, p_3$ as her rank-1, rank-2 and rank-3 posts respectively and the two posts
$p_6$ and $p_7$ are tied as her
rank-4 posts. For every agent $a$, the posts which are bold denote the set $f(a)$, whereas the posts which are
underlined denote the set $s(a)$. The instance $G$ admits a popular matching;  $M$ and
$M'$ shown  below are both  popular in $G$.
\begin{eqnarray}
M = \{(a_1, p_6), (a_2, p_1), (a_3, p_8), (a_4, p_2), (a_5, p_3), (a_6, p_9), (a_7, p_4)\} \label{eq:M} \\
M' = \{(a_1, p_6), (a_2, p_1), (a_3, p_8), (a_4, p_2), (a_5, p_4), (a_6, p_3), (a_7, p_5)\} \label{eq:M'}
\end{eqnarray}
Figure~\ref{fig:example-ex2}(b) shows the switching graph $G_M$ with respect to the
popular matching $M$. We note that the edges $(a_4, p_3)$ and $(a_1, p_1)$ get deleted
in Step~4 and Step~9 of Algorithm~\ref{algo:pop-matching}, respectively. Hence the switching
graph $G_M$ does not have the edges $(M(a_4) = p_2, p_3)$ and $(M(a_1) = p_6, p_1)$ respectively.
Consider the switching path $T = \langle p_9, p_3, p_4, p_5 \rangle$ in $G_M$.
By {\em applying} $T$ to $M$ we get
$M' = M \cdot T$ (see Equation~(\ref{eq:M'}))
which is also popular in $G$.
\begin{figure}[ht]
\begin{minipage}[b]{0.45\linewidth}
\centering
\begin{equation*}
\setlength{\arraycolsep}{0.5ex}\setlength{\extrarowheight}{0.25ex}
\begin{array}{@{\hspace{1ex}}c@{\hspace{1ex}}
              @{\hspace{1ex}}c@{\hspace{1ex}}
              @{\hspace{1ex}}c@{\hspace{1ex}}
              @{\hspace{1ex}}c@{\hspace{1ex}}
              @{\hspace{1ex}}c@{\hspace{1ex}}
              @{\hspace{1ex}}c@{\hspace{1ex}}}
    a_{1}: & {\bf p_1} \ & p_2 \ & p_3 \ & (\underline{p_6}, \underline{p_7}) \ &  \\[.5ex] 
    a_{2}: & \bf{p_1} \ & p_2 \ & \underline{p_8}  \ &\\[.5ex] 
    a_{3}: & \bf{p_1} \ & \underline{p_8} \ &{}  \ & {}\\[.5ex] 
    a_{4}: & \bf{(p_2, p_3)} \ &  p_1 \ & \underline{p_8} \ &   \ & {}\\[.5ex] 
    a_{5}: & \bf{p_3} \ &  (p_2, \underline{p_4}) &   \ & {}\\[.5ex] 
    a_{6}: & \bf{p_3} \ & \underline{p_9} \ & p_1   \ & {}\\[.5ex] 
    a_{7}: & \bf{(\underline{p_4}, \underline{p_5})} \ & p_1 \ &   \ & {}\\[.5ex] 
\end{array}
\end{equation*}
(a)
\end{minipage}
\begin{minipage}[b]{0.45\linewidth}
\centering

\psfrag{p1}{$p_1$}
\psfrag{p2}{$p_2$}
\psfrag{p3}{$p_3$}
\psfrag{p4}{$p_4$}
\psfrag{p5}{$p_5$}
\psfrag{p6}{$p_6$}
\psfrag{p7}{$p_7$}
\psfrag{p8}{$p_8$}
\psfrag{p9}{$p_9$}
\psfrag{x}{$-1$}
\psfrag{y}{$+1$}
\psfrag{z}{$0$}

\includegraphics[scale=0.4]{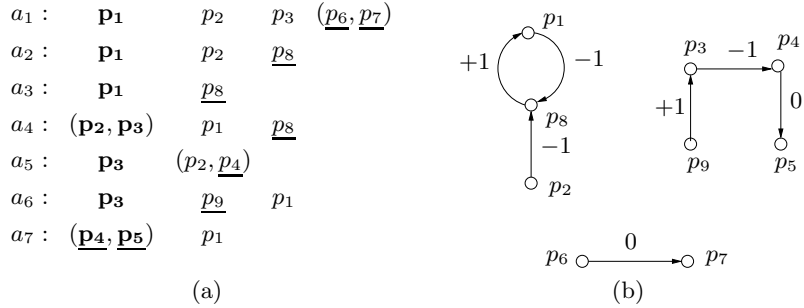}

(b)
\end{minipage}
\caption{(a) Preference lists of agents $\{a_1, \ldots, a_7\}$. The posts which are bold denote $f(a)$ and the
posts which are underlined denote $s(a)$.
(b) Switching graph $G_M$ with respect to the popular matching $M$ in $G$.}
\label{fig:example-ex2}
\end{figure}

\subsection{Some useful properties}
\label{sec:properties}
In this section we prove some useful properties of the switching graph $G_M$. 
Recall that the vertices $\A \cup \p$ are partitioned as $\oddone, \evenone, \unone$ w.r.t.
a maximum matching $M_1$ in $G_1$ (see Step~3 of Algorithm~\ref{algo:pop-matching}).
Further, the vertices $\A \cup \p$  are partitioned as $\oddtwo, \eventwo, \untwo$ w.r.t.
a popular matching $M$ in $G'$ (see Step~8 of Algorithm~\ref{algo:pop-matching}).

\begin{property}\label{prop0}
All sink vertices of $G_M$ belong to the set $\evenone$.
\end{property}
\begin{proof}
Assume for the sake of contradiction that $p$ is a sink vertex in $G_M$ and $p \in \oddone \cup \unone$. Recall that the
sink vertices of $G_M$ are unmatched posts in the popular matching $M$ in $G$. Since $M$ is a popular matching, it implies
that $M$ is a maximum matching on rank-1 edges in $G$. However, every maximum matching on rank-1 edges of $G$ matches
every vertex in $\oddone \cup \unone$. Thus, if $p$ is unmatched in $M$ and $p \in \oddone \cup \unone$,
it implies that $M$ is not a maximum matching on rank-1 edges of $G$, a contradiction.
\end{proof}

\begin{property}\label{prop1}
Every post $p$ belonging to a sink component has a path to a sink and hence belongs to the set $\eventwo$.
Every post belonging to a non-sink component belongs to the set $\untwo$.
\end{property}
\begin{proof}
We prove that a post $p$ belongs to a sink component of $G_M$ iff $p \in \eventwo$.
Let $p$ be a post such that $p \in \eventwo$. Then $p$ is either unmatched in $M$
or $p$ has an even length alternating
path starting at an unmatched vertex $p'$ with respect to $M$ in $G'$. If $p$ is
unmatched, then $p$ is a sink vertex in $G_M$ and hence we are done. Else let
$\langle p=p_1, a_1, \ldots, p_k, a_k, p_{k+1}= p'\rangle$ denote the alternating
path  and for every $1 \le i \le k$, we have
 $M(p_i) = a_i$. Note that every unmatched edge $(a_i, p_{i+1})$
is of the form $\oddtwo \eventwo$ and hence none of these unmatched edges
get deleted in Step~9 of Algorithm~\ref{algo:pop-matching}.
Therefore, it is easy to see that the path $\langle p=p_1, p_2, \ldots, p_{k+1}=p'\rangle$ is
present in $G_M$ and hence $p$ belongs to the sink component that contains $p'$.

To prove the other direction
let $\X$ be a sink component in $G_M$ and let $p$ be some vertex in $\X$.
If $p \in \eventwo$, it is easy to see that $p$ has a path to some sink vertex in $\X$
and we are done, else $p \in \untwo$. Recall
that $\oddtwo \cap \p = \emptyset$. Since $p$ belongs to $\X$, there exists some vertex
$p'$ in $\X$ such that
$p' \in \eventwo$ and $p'$ has a path to $p$ in $G_M$. Let $\langle p'= p_1, p_2, \ldots, p_k= p\rangle$
denote a path of minimal length. Note that for each $1\le i \le k$, $p_i$ is matched in $M$ and let $M(p_i) = a_i$.
By the minimality of the path, we know that for $2 \le i \le k$, $p_i \in \untwo $.  However,  $p_1 \in \eventwo $
implies $a_1 \in \oddtwo$. Thus the presence of edge $(p_1, p_2)$ in $G_M$ implies that there is a $\oddtwo \untwo$ edge in the graph $G''$
which should have been deleted by step~9 of Algorithm~\ref{algo:pop-matching}. Hence such an edge cannot be present in $G_M$
contradicting the fact that $p \in \untwo$.
This implies that every vertex contained in a sink component of $G_M$ has a path to some sink in the component.

The above proof immediately implies that a post $p$ belongs to a non-sink component iff $p \in \untwo$.
This finishes the proof of Property~\ref{prop1}.
\end{proof}

\begin{property} \label{prop-table-weights}
For an edge $(p_i, p_j)$ in $G_M$, the weight  $w(p_i, p_j)$ is determined by which partition
$p_i$ and $p_j$ belong to when 
 vertices are partitioned 
as $\oddone, \evenone, \unone$.
$w(p_i, p_j)$ can be determined using Table~\ref{tab:edge-weights}.

\begin{table}[ht]
\centering
\begin{tabular}{ |c || c | c | c |}
\hline
\backslashbox{$p_i$}{$p_j$}   & $\oddone$ & $\evenone$ & $\unone$ \\
\hline
$\oddone$ & $0$ & $-1$ & $\times$ \\
\hline
$\evenone$ & $+1$ & $0$ & $\times$ \\
\hline
$\unone$ & $\times$ & $-1$ &  $0$\\
\hline
\end{tabular}
\caption{Table shows $w(p_i, p_j)$ for an edge $(p_i, p_j)$ in $G_M$. The weight is determined by the partition of vertices as $\oddone, \evenone, \unone$. The $\times$ denotes that
such an edge is not present in $G_M$.}
\label{tab:edge-weights}
\end{table}
\end{property}

\begin{proof}
To prove Property~\ref{prop-table-weights} we justify the entries in Table~\ref{tab:edge-weights}.
Let $(p_i, p_j)$ be an edge in $G_M$ and let $M(p_i) = a$.
The weight on the edge $(p_i, p_j)$ is determined by the relative ranks of $p_i$ and $p_j$ in $a$'s preference list.
We note that a post $p \in \oddone \cup \unone$ has only rank-1 edges incident
on it in the graph $G'$.
Hence if  $p_i \in \oddone \cup \unone$, then $a$ treats $p_i$ as her rank-1 post.
\begin{itemize}
\item  $p_i \in \oddone$: $a$ treats $p_i$ as her rank-1 post and since posts in $\oddone$ remain matched along agents in $\evenone$, it
implies that $a \in \evenone$.
\begin{itemize}
\item $p_j \in \oddone$: $a$ treats $p_j$ as her rank-1 post, thus, $w(p_i, p_j) = 0$.
\item $p_j \in \evenone$: We show that $a$ treats $p_j$ as her non-rank-1 post and hence
$w(p_i, p_j) = -1$. Assume for the sake of contradiction that $a$ treats $p_j$ as a rank-1 post.
It implies that there is a $\evenone \evenone$ edge
in the graph $G_1$, a contradiction (by part (c) Lemma~\ref{lemma:node-classification}).
\item $p_j \in \unone$: We show that such an edge cannot exist in $G_M$. Recall
that posts in $\unone$ have only rank-1 edges incident on them, hence $a$ treats $p_j$ as her rank-1 post.
This implies that there is a $\evenone \unone$ edge in $G_1$ a contradiction (by part (c) of Lemma~\ref{lemma:node-classification}).
\end{itemize}
\item  $p_i \in \evenone$: Here we consider two cases:

(i) $a$ treats $p_i$ as her rank-1 post: In this case, we note that $s(a) \subseteq f(a)$ and hence
$a$ has only rank-1 edges incident on it in the graph $G'$ and all these edges are incident on posts which belong to $\evenone$.
Thus the only case possible is, $p_j \in \evenone$ and $w(p_i, p_j) = 0$.

(ii) $a$ treats $p_i$ as her non-rank-1 post: We first note that $a \in \evenone$ because agents in $\oddone \cup \unone$ remain
matched along rank-1 edges in every popular matching. Consider the three different cases for $p_j$.

\begin{itemize}
\item $p_j \in \oddone$: $a$ treats $p_j$ as her rank-1 post and hence $w(p_i, p_j) = +1$.
\item $p_j \in \evenone$: We show that $a$ treats $p_j$ as her non-rank-1 post and hence $w(p_i, p_j) = 0$. Assume
for the sake of contradiction that $a$ treats $p_j$ as her rank-1 post. Then there exists an $\evenone \evenone$ edge
in $G_1$ a contradiction (by part (c) of Lemma~\ref{lemma:node-classification}).
\item $p_j \in \unone$: We show that such an edge cannot exist in $G_M$. If such an edge exists there is an $\evenone \unone$ edge
in $G_1$ a contradiction (by part (c) of Lemma~\ref{lemma:node-classification}).
\end{itemize}

\item  $p_i \in \unone$: $a$ treats $p_i$ as her rank-1 post and since posts in $\unone$ remain matched along agents in $\unone$, it
implies that $a \in \unone$.
\begin{itemize}
\item $p_j \in \oddone$: $a$ treats $p_j$ as her rank-1 post however such an edge gets
deleted as an $\oddone \unone$ edge in Step~4 of Algorithm~\ref{algo:pop-matching}. Thus such an edge cannot be present in $G_M$.
\item $p_j \in \evenone$:  We show that $a$ treats $p_j$ as her non-rank-1 edge and hence $w(p_i, p_j) = -1$.
Assume for the sake of contradiction that $a$ treats $p_j$ as a rank-1 post then, it implies that there is a $\unone \evenone$ edge
in the graph $G_1$, a contradiction (by part (c) of  Lemma~\ref{lemma:node-classification}).
\item $p_j \in \unone$: $a$ treats $p_j$ as her rank-1 post and therefore $w(p_i, p_j) = 0$.
\end{itemize}
\end{itemize}

\end{proof}

\begin{property}\label{prop2}
Every path $T$ in $G_M$ has  $w(T) \in \{-1, 0, +1\}$. Every cycle $C$ in $G_M$ has  $w(C) = 0$. There exists no
path $T$ in $G_M$ ending in a sink vertex with $w(T) = +1$.
\end{property}
\begin{proof}
It is easy to observe that if the edges have weights according to Table~\ref{tab:edge-weights}, then
every path in $G_M$ has weight belonging to $\{-1, 0, +1\}$. Further every cycle has to have
weight $0$. It remains to argue that in $G_M$ there exists no  path $T$ of weight $+1$ which ends in a sink.
For contradiction, assume that such a path exists in $G_M$ and consider
applying the path $T$ to $M$ to obtain the matching $M' = M \cdot T$. The number of agents
that prefer $M'$ to $M$ is exactly one more than the number of agents that prefer $M$ to $M'$. Thus $M'$ is more
popular than $M$, contradicting the fact that $M$ was a popular matching.
\end{proof}

\begin{property}\label{prop3}
For any switching path $T$ (or switching cycle $C$) in $G_M$, the matching $M'=M \cdot T$ ($M' = M \cdot C$ resp.) is a popular matching in $G$.
Every popular matching $M'$ in $G$ can be obtained from $M$ by applying to $M$ one or more vertex disjoint switching paths and switching cycles
in each of a subset 
of sink components of $G_M$ together with one or more vertex disjoint switching cycles in each of a subset of the non-sink components of $G_M$.
\end{property}
\begin{proof}
Let $T =  \langle p_1, p_2, \ldots, p_k\rangle$ be a switching path in $G_M$ with $p_k$ unmatched in $M$
and $M' = M \cdot T$ denote the matching obtained by applying the path $T$ to $M$.
Let $\A_T = \cup_{i=1}^{k-1} \{M(p_i) = a_i\}$.
We observe that for every $a_i \in \A_T$, $M'(a_i) \in \{f(a_i) \cup s(a_i)\}$ because edges of $G_M$
are derived from a subset of graph $G' = (\A \cup \p, E')$ (refer to Algorithm~\ref{algo:pop-matching}). Further for any $a \notin \A_T$,
$M'(a) = M(a)$. Finally note that since $w(T) = 0$, for every agent that got demoted from her rank-1 post
there exists a
unique agent who got promoted to her rank-1 post in $M'$. Thus, $M'$ is
a maximum matching on rank-1 edges of $G$. It is therefore clear that $M'$ satisfies both the properties defined by
Theorem~\ref{thm:pop-mat} and hence $M'$ is a popular matching in $G$.
An exactly similar argument proves that for any switching cycle $C$, $M \cdot C$ is also a popular matching in $G$.

Let $M'$ be any popular matching in $G$.
Consider $M \oplus M'$, this is set of vertex disjoint even length paths and even length cycles in $G$.
Let $T_G =\langle p_1, a_1, \ldots,p_k, a_k, p_{k+1}\rangle$ be any even length path in $M \oplus M'$  with $p_{k+1}$
unmatched in $M$ and $p_1$ unmatched in $M'$.
For every $1 \le i \le k$, let
 $M(p_i) = a_i$.
Note that every unmatched edge $(a_i, p_{i+1})$
is of the form $\oddtwo \eventwo$ and hence none of these unmatched edges
get deleted in Step~9 of Algorithm~\ref{algo:pop-matching}.
Therefore, it is easy to see that the path $\langle p=p_1, p_2, \ldots, p_{k+1}=p'\rangle$ is
present in $G_M$ and it ends in a sink.
 Note that $w(T)$ cannot be strictly positive since $M$ is a popular matching. Similarly,
$w(T)$ cannot be strictly negative. This is because since both $M$ and $M'$ are popular, $w(T) \le -1$
implies that there exists another path $T'_G$ (or a cycle $C'_G$) in $M \oplus M'$, whose corresponding path $T'$ (resp. cycle $C'$)
in the graph $G_M$ has a positive weight. However, this again contradicts the fact that $M$ is a popular matching.
Thus, the path $T$ has weight $0$ and ends in a sink and hence is a switching path. A similar argument shows
that every cycle in $M \oplus M'$ has a corresponding switching cycle in $G_M$. Applying these switching paths
and cycles to $M$ gives us the desired matching $M'$, thus completing the proof.
\end{proof}

Recall the definition of $choices(a)$ for an agent as given by Definition~\ref{def:choices}.
We now define the notion of a {\em \tightpair}, that is, a set of agents $\A_1$ and a set
of posts $\p_1$ with $|\A_1| = |\p_1|$. Further, for every $a \in \A_1$ we have $ choices(a) \subseteq \p_1$.
We 
show that a \tightpair exists whenever there is a non-sink component in the
switching graph $G_M$.
\begin{lemma}
\label{lem:tight-sets}
Let $\Y$ be a  non-sink component  in $G_M$ and $q \in \Y$.
Let,
\begin{center}
 $\p_{q} = {q} \cup \{p:  \mbox{ $q$ has a path to $p$ in $G_M$ } \}$
\end{center}
Then there exists a set of agents $\A_{q}$ such that
(i) $|\A_q| = |\p_q|$, and
(ii) for every $a \in \A_q$,  $choices(a) \subseteq \p_q$.
\end{lemma}
\begin{proof}
Let $\A_q = \cup_{p \in \p_q} M(p)$.
Since every $p \in \p_q$ is matched, we note that $|\A_q| = |\p_{q}|$.
Consider any agent $a \in \A_q$;  then $M(a) \in \p_q$ and note that $M(a) \in choices(a)$. Further,
note that, for every $p' \in choices(a) \setminus \{M(a)\}$, we have an edge $(M(a), p')$ in $G_M$. Thus,
every such $p'$ also belongs to $\p_q$. This proves that for every $a \in \A_q$, $choices(a) \subseteq \p_q$.
\end{proof}
\subsection{Generating popular pairs and counting popular matchings}
\label{sec:gen-pop-pairs}
Let $G = (\A \cup \p, E)$ be an instance of
the popular matchings problem. Define
\begin{eqnarray}
\label{eqn:pop-pairs}
PopPairs = \{ (a, p) \in E: \mbox { $M$ is a popular matching in $G$ and $M(a) = p$} \}.
\end{eqnarray}
Using the switching graph defined in the previous section, it is easy to compute the set $PopPairs$
in $G$. Let $G_M$ be the switching graph with respect to a popular matching $M$ in $G$. 
From \textit{Property~\ref{prop3}} we can conclude that an edge $e = (a, p)$ is a popular
pair if and only if (i) $e \in M$  or, (ii) the edge $(M(a), p)$ belongs to some switching path in $G_M$ or, (iii)
the edge $(M(a), p)$ belongs to some switching cycle in $G_M$.

We note that edges satisfying condition (i) can be marked in $O(\sqrt{n}m)$ time using Algorithm~\ref{algo:pop-matching}
and edges satisfying conditions (ii) and (iii) can be marked in linear time in the size of the switching graph.
Thus, we conclude the following theorem.
\begin{theorem}
\label{thm:pop-pairs}
The set of popular pairs for an instance $G = (\A \cup \p, E)$ of the popular matchings problem 
with ties can be computed in $O(\sqrt{n}m)$ time. 
\end{theorem}
\begin{proof}
From \textit{Property~\ref{prop3}} we can conclude that an edge $e = (a, p)$ is a popular
pair if and only if (i) $e \in M$  or, (ii) the edge $(M(a), p)$ belongs to some switching path in $G_M$ or, (iii)
the edge $(M(a), p)$ belongs to some switching cycle in $G_M$. We show how each of these
conditions can be efficiently verified.
\begin{itemize}
\item The condition (i) can be checked in $O(\sqrt{n}m)$ time by running Algorithm~\ref{algo:pop-matching} and obtaining
a popular matching $M$.
\item In order to check condition (iii), recall that every cycle in $G_M$ has weight $0$ and is therefore a
switching cycle. Hence this condition can be checked in linear time in the size of the switching graph
by identifying strongly connected components of $G_M$. Every edge belonging to a strongly connected component is a popular pair.
\item In order to check condition (ii), recall
that a switching path is a path which has weight $0$ and ends in a sink. Therefore such paths can be found only in
sink components of $G_M$ or equivalently paths beginning at vertices in $\eventwo$. Further,
any sink vertex in $G_M$ has to be a vertex in $\evenone$ according to the partition on rank-1 edges of $G$.
Using the weights on the edges given by Table~\ref{tab:edge-weights}, it is easy to see that
any $0$ weight path ending in  a sink has to begin at a vertex $p \in \evenone$. Thus,
a simple depth-first search beginning at vertices in $\evenone \cap \eventwo$ and marking all edges that we encounter
as popular pairs takes care of condition (ii). It is easy to see that this procedure also takes time linear in the size of $G_M$.
\end{itemize}
\end{proof}

We now show that given an instance of the popular matchings problem with ties,
the problem of counting the number of popular matchings is 
$\#$P-Complete. 
\begin{theorem}
\label{thm:counting-hard}
Given an instance $G = (\A \cup \p, E)$ of the popular matchings problem with ties, counting the total
number of popular matchings in $G$ is $\#$P-Complete.
\end{theorem}
\begin{proof}
In order to prove the completeness
result, we reduce from the problem of counting the number of perfect matchings in $3$-regular
bipartite graphs. This problem was shown to be $\#$P-Complete by Dagum and Luby~\cite{DL92}.
Let $H = (\A \cup \p,E )$ be a $3$-regular bipartite graph. We construct an instance $G = H$ of the popular
matching problem by assigning all the edges in $E$ as rank-1 edges. It is well-known that a $k$-regular bipartite
graph admits a perfect matching and hence it is easy to see that every perfect matching in $H$ is a popular
matching in $G$ and vice versa. Thus, the theorem statement follows.
\end{proof}

\section{Cheating strategies -- preliminaries}
\vspace{-0.1in}
\label{sec:cheating-prelims}

In this section we set up the notation useful
in formulating our cheating strategies.
We begin by partitioning the set of agents $\A$ depending on the posts that a particular
agent gets matched to when each  agent is truthful,
that is, in the instance $G$.
\begin{eqnarray*}
&\mbox { \fonly } &= \{a: \mbox {every popular matching in $G$ matches $a$ to one of her rank-1 posts} \}\\
&\mbox { \sonly } &= \{a: \mbox {every popular matching in $G$ matches $a$ to one of her non-rank-1 posts} \}\\
&\mbox { \fands } &= \A \setminus ( \mbox { \fonly } \cup \mbox { \sonly } ).
\end{eqnarray*}
The set $\A_{f/s}$ denotes the set of agents $a$ such that $a$ gets matched to one of her rank-1 posts in some popular matching in
$G$, whereas to one of her non-rank-1 posts in some other popular matching in $G$.
It is easy to see that the above partition can be readily obtained once we have the
set of popular pairs $PopPairs$ (defined by Equation~(\ref{eqn:pop-pairs})).

Let $a_1$ be the sole manipulative agent who is aware
of the true preference lists of all other agents.
Let $\mathcal{L} = P_1, P_2, \ldots, P_t, \ldots, P_l$ denote the true preference list of $a_1$ where $P_i$ denotes the set of
$i$-th rank posts of $a_1$.
Since we will be working with another instance $H$ obtained by falsifying the preference list of $a_1$,
we now qualify the sets
$f(a)$ and $s(a)$ for every agent with the instance under consideration.
For an agent $a$, let $f_G(a)$ and $s_G(a)$ denote sets $f(a)$ and $s(a)$ respectively for agent $a$ in $G$.
We note that $f_G(a_1) = P_1$. Assume that
$s_G(a_1) \subseteq P_t$ is the set of $t$-th ranked posts of $a_1$, where $t > 1$.

Recall the strategy -- {\em better always}
defined for a single manipulative agent.
If agent $a_1 \in \A_f$, then she does not have any incentive to manipulate
her preference list. Thus, in this case we are done and $\mathcal{L}$ is
her optimal strategy. We therefore focus on $a_1 \in \A_s \cup \A_{f/s}$.
Let $H$ denote the instance obtained by falsifying the preference list of $a_1$ alone.
\begin{itemize}
\item If $a_1 \in \A_s$, then in order to get {\em better always} her goal
is to force at least some popular matching in $H$ to match her to a post
which she strictly prefers to her $t$-th ranked post (that is,
$s_G(a_1)$).
\item If $a_1 \in \A_{f/s}$, then in order to get {\em better always}
her goal is to force every popular matching in $H$ to match her to one of her true
rank-1 posts.
\end{itemize}

Denote by $H\succ G$ with respect to $a_1$ if agent $a_1$ is {\em better always}
in $H$. It is instructive to consider examples to get intuition regarding the
cheating strategies.

\begin{example}
\label{app:ex:sonly}
\end{example}
Consider the instance $G$ as shown in Figure~\ref{fig:example-ex2}
and let $a_5$ be the manipulative agent.
It can be seen that $a_5 \in \A_{f/s}$ in $G$. Now consider the instance $H$ where $a_5$ alone
falsifies her preference list such that
 $p_3$ is her rank-1 post and $p_8$ as her rank-2 post.
\begin{center}
$a_5: p_3 \hspace{0.1in} p_8$
\end{center}
It is easy to
verify that every popular matching in $H$ matches $a_5$ to $p_3$ which is her true rank-1 post.
The idea for an $\A_{f/s}$ agent $a$ is to choose a post in $s_H(a)$ (here $p_8$) to which $a$ can
never be matched in a popular matching in $H$. We will show that such a post can be chosen whenever there exists a non-sink component in
the switching graph and therefore a  \tightpair   (in this case
$\p_1 = \{p_8, p_1\}$ and $\A_1 = \{a_2, a_3\}$).

\begin{example}
\label{ex:need-modified}
\end{example}
Consider the instance $G$ shown in Figure~\ref{fig:example-ex2} and
let $a_1$ be the manipulative agent.
Every popular matching in $G$ matches $a_1$ to either $p_6$ or $p_7$ and therefore $a_1 \in \A_s$. 
Let $H$ denote the instance where $a_1$ submits the preference list as follows:
$p_3$ is her
rank-1 post whereas $p_8$ is her rank-2 post.
\begin{center}
$a_1: p_3 \hspace{0.1in} p_8$
\end{center}
It can be verified that every
popular matching in $H$ matches $a_1$ to $p_3$. The intuition here is that, a post to which $a_1$ wishes to get matched
(here $p_3$), should have a path to an unmatched post or roughly belong to
a sink component of $G_M$.  We also choose a post in $s_H(a_1)$ (in this case $p_8$) to which $a_1$ can never get
matched in any popular matching in $H$.
However, in this example, this is not the best that $a_1$ can get by falsifying. Let $a_1$ falsify her preference list as below
and let $H$ denote the falsified instance.
\begin{center}
$a_1: p_2 \hspace{0.1in} p_8$
\end{center}
Consider the matching $M'' = \{(a_1, p_2), (a_2, p_1), (a_3, p_8), 
(a_4, p_3), (a_5, p_4), (a_6, p_9), (a_7, p_5)\}$ in $H$. It can be verified that
$M''$  is popular in $H$  and in fact
every popular matching in $H$ matches $a_1$ to $p_2$. However,
our intuition that $p_2$ should belong to a sink component does not hold. This
is because the edge $(a_4, p_3)$ which got deleted in Step~4 of Algorithm~\ref{algo:pop-matching}
is being used after $a_1$ falsifies her preference list.
In order to deal with such cases we will work with a slightly modified instance as
defined in Section~\ref{sec:modified-instance}.

\subsection{$s(a)$  for other agents remains unchanged}
\label{subsec:sposts-dontchange}
Let $H$ denote the instance obtained by falsifying the preference list of $a_1$ alone.
Since the rest of the agents are truthful,
for every agent $a \in \A \setminus \{a_1\}$, we have $f_H(a) = f_G(a)$.
However, since $s_H(a)$ depends on the rank-1 posts of the rest of the agents,
it may
be the case that when $a_1$ falsifies her preference list, $s_H(a) \neq s_G(a)$  for an agent $a \in \A \setminus \{a_1\}$.
We claim that if $a_1$ falsifies her preference list only to improve the rank of the post that
she gets matched to, the rest of the agents do not change their $s(a)$.
Recall
that by definition, $s_H(a)$ is the set of most preferred posts of $a$ which are {\em even} in the graph $H_1$ (the graph
$H$ on rank-1 edges).
Theorem~\ref{lem:partition-invariant} states the claim; its proof requires the following lemmas.

\begin{lemma}
Let $a_1 \in \A_s \cup \A_{f/s}$ when she is truthful. Then, $f_G(a_1) \subseteq (\oddone)_G$.
\label{lem:f-in-odd1}
\end{lemma}
\begin{proof}
Suppose not. Then, let $q \in f_G(a_1)$ such that, $q \in (\evenone \cup \unone)_G$. This implies
that $a_1 \in (\oddone \cup \unone)_G$. Note that any agent in $(\oddone \cup \unone)_G$ remains matched
along her rank-1 edge in every popular matching of $G$. This is because if $a_1 \in (\oddone)_G$, then
$s_G(a_1) \subseteq f_G(a_1)$ and therefore $a_1$ has no non-rank-1 edges incident on it in the graph $G'$.
On the other hand, if $a_1 \in (\unone)_G$ and if $a_1$ gets matched to a non-rank-1 post in a popular
matching $M$, then $M$ is not a maximum matching on rank-1 edges of $G$.
Therefore in each case, $a_1$ remains matched along a rank-1 edge in every popular matching in $G$
contradicting the fact that $a_1 \in \A_s \cup \A_{f/s}$.
Thus, $f_G(a_1) \subseteq (\oddone)_G$.
\end{proof}

\begin{lemma}
Let $H$ be such that $H \succ G$ w.r.t. $a_1$. Then $f_H(a_1) \subseteq (\oddone \cup \unone)_G$.
\label{lem:f-in-ou}
\end{lemma}
\begin{proof}
Assume for contradiction that there exists an instance
$H \succ G$ w.r.t $a_1$ and let $f_H(a_1) \cap (\evenone)_G = \{q_1, \ldots, q_k\}$.
Note that the rank of each $q_i$ in $a_1$'s preference
list is at least $t$.
We show that every popular matching in $H$ matches $a_1$ to one of $q_i$ and hence the rank
of the most preferred post that $a_1$ gets in $H$ is at least $t$,
a contradiction to  $H \succ G$ w.r.t. $a_1$.

We first show that if $f_H(a_1) \cap (\evenone)_G \neq \phi$, then the size of the
maximum matching on rank-1 edges of $H$ is strictly larger than the size of the maximum
matching on rank-1 edges of $G$.
Let $G_1$ be the graph on rank-1 edges of $G$ and let $M_1$ be a maximum matching in $G_1$ that leaves $a_1$ unmatched.
Note that such a matching exists because $f_G(a_1) \subseteq (\oddone)_G$ which implies that $a_1 \in (\evenone)_G$.
Consider the graph $H_1$, that is, the graph on rank-1 edges of $H$. Note that $M_1$ is a matching in $H_1$ as no other agent
changes her preference list.
Since each $q_i \in (\evenone)_G$ and $a_1 \in (\evenone)_G$, the addition of the edge $(a_1, q_i)$ creates an augmenting
path with respect to $M_1$ in the graph $H_1$. Thus, we get another matching, say $M_2$, in $H_1$ such that $|M_2| = |M_1 + 1|$.

Now consider a popular matching $M'$ in $H$ and let $M_1'$ denote the matching $M'$ restricted
to rank-1 edges of $H$.
Since $M'$ has to be a maximum matching on rank-1 edges of $H$, it is clear that $|M_1'| = |M_2|$.
Further since $M'(a_1) \in \{f_H(a_1) \cup s_H(a_1)\}$ let us consider the following cases:
\begin{itemize}
\item $M'(a_1) \in \{q_1, \ldots, q_k\}$: In this case we have the desired contradiction and we are done.
\item $M'(a_1) = q$ where $q \in  f_H(a_1) \cap (\oddone \cup \unone)_G$: If $q \in f_G(a_1)$,
then the edge $(a_1, q) \in G_1$ and therefore $M_1'$ is in fact a matching in $G_1$. Now
since $|M_1'| = |M_1|+1$, it contradicts the fact that $M_1$ was a maximum matching in $G_1$. Therefore assume that,
$q \notin f_G(a_1)$. Note that $M_2' = M_1' \setminus \{(a_1, q)\}$ is in fact a maximum
matching of $G_1$ since no other agents changed their preferences. However, $M_2'$ leaves $q$ unmatched
which contradicts the fact that $q \in (\oddone \cup \unone)_G$, since every vertex in
$(\oddone \cup \unone)_G$ remains matched in any maximum matching of $G_1$.
\item $M'(a_1) \in s_H(a_1)$: If $s_H(a_1) \subseteq f_H(a_1)$, the previous cases have already handled this.
In the other case assume that $s_H(a_1) \cap f_H(a_1) = \phi$. This implies that $M_1'$ leaves $a_1$ unmatched.
Thus, $M_1'$ is also a matching in $G_1$
since no other agents changed their preferences. However, $|M_1'| = |M_1| + 1$ which contradicts the fact that
$M_1$ was a maximum matching in $G_1$.
\end{itemize}
This finishes the proof of the lemma.
\end{proof}

\begin{lemma}
\label{lem:max-inG1-max-inH1}
Let $M_1$ be a maximum matching in $G_1$ such that $M_1$ leaves $a_1$ unmatched. Then,
in any instance $H$ such that $H \succ G$ w.r.t. $a_1$,
$M_1$ is a maximum matching in $H_1$.
\end{lemma}

\begin{proof}
We first note that such a maximum matching $M_1$ in $G_1$ which leaves $a_1$ unmatched exists, because
$f_G(a_1) \subseteq (\oddone)_G$, hence $a_1 \in (\evenone)_G$. Assume that $M_1$ is not a maximum matching
in $H_1$.
Then there exists an augmenting path $\langle a_1, p_1, \ldots, a_k, p_k\rangle $ in $H_1$ with respect to $M_1$ where
both $a_1$ and $p_k$ are unmatched.
However, using the path $\langle p_k, a_k, \ldots, p_1 \rangle$,
we have an even length alternating path from $p_k$ to $p_1$ contradicting the fact that $p_1 \in (\oddone \cup \unone)_G$.
Thus $M_1$ is a maximum matching in $H_1$.
\end{proof}

\begin{theorem}
\label{lem:partition-invariant}
Let $H$ be an instance such that $H \succ G$ w.r.t. $a_1$. Then, (i) $(\evenone)_G \cap \p = (\evenone)_H \cap \p$
and therefore $s_H(a) = s_G(a)$ for every $a \in \A \setminus \{a_1\}$ and, (ii) $(\oddone)_G \cap \A = (\oddone)_H \cap \A$.
\end{theorem}

\begin{proof}
The proof uses Lemma~\ref{lem:f-in-odd1}, Lemma~\ref{lem:f-in-ou}, and Lemma~\ref{lem:max-inG1-max-inH1} proved above.
The case when $f_H(a_1) = f_G(a_1)$ is easy, since $H_1 = G_1$ and both (i) and (ii) are trivially true.
Consider the case when
$f_H(a_1) \neq f_G(a_1)$ and let $M_1$ be a maximum matching in $G_1$ such that $M_1$ leaves $a_1$ unmatched.
By Lemma~\ref{lem:max-inG1-max-inH1},
$M_1$ is also a maximum matching in $H_1$. To prove $(\evenone)_G \cap \p = (\evenone)_H \cap \p$, consider the
following two cases:
\begin{itemize}
\item $p \in (\oddone \cup \unone)_G \cap \p$: Assume for contradiction that $p \in (\evenone)_H \cap \p$.
This implies that there exists an even length
alternating path $T$ with respect to $M_1$ in $H_1$ from an unmatched post  in $H_1$. The path $T$
can not contain $a_1$, since $a_1$ is unmatched in $M_1$, thus $T$ is present in $G_1$
contradicting the fact that $p \in (\oddone \cup \unone)_G \cap \p$.
\item $p \in (\evenone)_G \cap \p$: Let $T$ denote the even length alternating path w.r.t. $M_1$
in $G_1$ starting at an unmatched post in $M_1$.
The path $T$ again can not contain $a_1$ and hence exists in $H_1$ thus proving that $p \in (\evenone)_H \cap \p$.
\end{itemize}
Now, since the preference lists of the agents $a \in \A \setminus \{a_1\}$ remain unchanged, it is clear that
$s_H(a) = s_G(a)$.

Finally, consider any $a \in (\oddone)_G \cap \A$. For contradiction, assume that $a \in (\evenone)_H \cap \A$. Then $a$ has
an even length alternating path $T$ w.r.t. $M_1$ in $H_1$ starting at an unmatched agent. This path
has to begin at $a_1$, otherwise, it was already present in $G_1$. However, the existence of the
path $T$ beginning at $a_1$ implies that there exists a post $q \in f_H(a_1)$ such that $q \in (\evenone)_G$.
This contradicts Lemma~\ref{lem:f-in-ou} and thus finishes the proof of the lemma.

\end{proof}

\subsection{An $\A_s$ agent cannot get one of her true rank-1 posts}
\label{sec:sonly-cant-get-rank1}
In this section we  show that if $a_1 \in \A_s$,
then by falsifying her preference list alone, she cannot get matched to one of her true
rank-1 posts in any popular matching in $H$. 
We prove it using Theorem~\ref{thm:sonly-no-f-post} which requires the following lemmas. 

\begin{lemma}
Let $a_1 \in \A_s$, and let $q \in  f_G(a_1)$. Then, $q$ belongs to a non-sink component  of $G_M$ and the
edge $(M(a_1), q)$ is not
contained in a cycle in $G_M$.
\label{lem:sonly-not-in-cycle}
\end{lemma}

\begin{proof}
We first show that the edge $(M(a_1), q)$ is not contained in a cycle in $G_M$. Assume for the sake of
contradiction that there exists a cycle $C$ in $G_M$ which contains the edge $(M(a_1), q)$. Since every cycle
in $G_M$ has a weight $0$, the cycle $C$ is a switching cycle and hence we get another popular
matching $M' = M \cdot C$ in which $a_1$ gets matched to $q$. Since $q \in f_G(a_1)$, this contradicts the
fact that $a_1 \in \A_s$.

We now show that every $q \in f_G(a_1)$ belongs to a non-sink component of $G_M$. Assume not. Then,
let there exist some $q \in f_G(a_1)$ such that $q$ belongs to a sink component, say $\X$ of $G_M$. In this
case we show that there exists a switching path $T$ beginning at $M(a_1)$ which uses the edge $(M(a_1), q)$. Using
$T$, we construct another popular matching $M' = M \cdot T$ where $a_1$ gets matched to $q$. Thus, we
get the desired contradiction as $a_1 \in \A_s$.

It remains to prove that the switching path $T$ exists.
Observe that the since $q$ belongs to a sink component $\X$, by {\it Property~\ref{prop1}} there exists a path $T_1$ from $q$ to
some sink $q'$ in $\X$. By \textit{Property~\ref{prop0}}, it is clear that $q' \in \evenone$ and from Lemma~\ref{lem:f-in-odd1},
we know that $q \in \oddone$. Using  Table~\ref{tab:edge-weights} of edge weights, it is clear that the path $T_1$ starting
a vertex in $\oddone$ and ending in a vertex in $\evenone$
has weight $w(T_1) = -1$. Finally note that,
$w(M(a_1), q) = +1$ since $M(a_1) \in s_G(a_1)$ and $q \in f_G(a_1)$.
Thus, we obtain the switching path $T = \langle M(a_1), q, T_1 \rangle$ which ends in the sink $q'$ and has $w(T) = 0$.
This completes the proof of the lemma.
\end{proof}

\begin{lemma}
\label{lem:a1-notin-tightset}
Let $a_1 \in \A_s$ and let $q \in f_G(a_1)$. Let $\p_q$ be defined as
\begin{center}
$ \p_q  = q \cup \{q': \mbox {there is a path from $q$ to $q'$ in $G_M$} \}$
\end{center}
Let $\A_q = \cup_{q' \in \p_q} M(q')$.
Then, $a_1 \notin \A_q$.
\end{lemma}
\begin{proof}
Note that since $a_1 \in \A_s$ and $q \in f_G(a_1)$, by Lemma~\ref{lem:sonly-not-in-cycle}, $q$ belongs to a non-sink component,
say $\Y$, of $G_M$.
If $M(a_1)$ does not belong to $\Y$,
then it is clear that $M(a_1) \notin \p_q$. Otherwise, let $M(a_1)$ belong to $\Y$.
In this case, we show that if there exists a path from
$q$ to $M(a_1)$ in $\Y$, then along with the edge $(M(a_1), q)$ there is a cycle in $G_M$ containing the
edge $(M(a_1), q)$ which is a contradiction by Lemma~\ref{lem:sonly-not-in-cycle}. Thus,
we need to show that the edge $(M(a_1), q)$ does not get deleted in either Step~4 or Step~9
of Algorithm~\ref{algo:pop-matching}. Since $M(a_1)$ belongs to a non-sink component, therefore $M(a_1) \in (\untwo)_G$
which implies that $a_1 \in (\untwo)_G$.
Further, since $q \in f_G(a_1)$, by Lemma~\ref{lem:f-in-odd1}, we know that $q \in (\oddone)_G$.
This implies that $a_1 \in (\evenone)_G$. Thus, from Claim~\ref{claim:no-del-Step9}, it is clear that no edge incident
on $a_1$ gets deleted in either Step~9 or Step~4 of Algorithm~\ref{algo:pop-matching}. Thus, $a_1 \notin \A_q$.
\end{proof}

\begin{lemma}
\label{lem:tight-set-H}
Let $a_1 \in \A_s$ and let $q \in f_G(a_1)$. Then, there exists sets $\A_q$ and $\p_q$ such that
$|\A_q| = |\p_q|$ and
for every $a \in \A_q$ we have $choices_H(a) \in \p_q$.
\end{lemma}
\begin{proof}
Since $a_1 \in \A_s$ and $q \in f_G(a_1)$, from Lemma~\ref{lem:sonly-not-in-cycle}, we know that $q$ belongs to
a non-sink component, say $\Y$, of $G_M$. Therefore, using Lemma~\ref{lem:tight-sets}, we know that there exists a \tightpair
$\A_q$ and $\p_q$ such that $|\A_q| = |\p_q|$ and for each $a \in \A_q$,
we have $choices_G(a) \subseteq \p_q$.
Further, Lemma~\ref{lem:f-in-odd1} implies that $q \in (\oddone)_G$. Now consider any post $q' \in \p_q$.
From the Table~\ref{tab:edge-weights}, it is easy to see that $q' \in (\oddone \cup \evenone)_G$, since
posts in $(\unone)_G$ cannot be reached starting at a post in $(\oddone)_G$. Therefore,
we note that every $a \in \A_q$ is such that $a \in (\evenone \cup \oddone)_G$. This is because
in any popular matching, posts in ($\oddone \cup \evenone$) remain matched to agents in ($\oddone \cup \evenone$).
To prove the lemma it suffices to show that for every $a \in \A_q$, $choices_H(a) \subseteq choices_G(a)$.

By Lemma~\ref{lem:a1-notin-tightset}, $a_1 \notin \A_q$ and therefore,
we know that $\{f_G(a) \cup s_G(a) \} = \{f_H(a) \cup s_H(a) \}$, for
all $a \in \A_q$. If no edges had got deleted in Step~4 and Step~9 of
Algorithm~\ref{algo:pop-matching} then we would be done since $choices_G(a)$
would have been $\{f_G(a) \cup s_G(a) \}$. Observe that since every $a \in \A_q$ belongs to
a non-sink component of $G_M$, therefore $a \in (\untwo)_G$, thus, by Claim~\ref{claim:no-del-Step9}
no edge incident on $a$ gets deleted in Step~9 of Algorithm~\ref{algo:pop-matching} when run on $G$.

It remains to show that if $a' \in \A_q$ and edge $(a',q')$ gets deleted in Step~4 of Algorithm~\ref{algo:pop-matching}
when run on $G$, then $(a', q')$  also gets deleted in Step~4 of Algorithm~\ref{algo:pop-matching} when run on $H$.
If $a' \in (\evenone)_G$, by Claim~\ref{claim:no-del-Step9} no edge incident
on $a'$ gets deleted in Step~4 of Algorithm~\ref{algo:pop-matching} when run on $G$. Finally,
let $a' \in (\oddone)_G$. If the edge $(a', q')$ got deleted in Step~4, then $q' \in (\oddone \cup \unone)_G$.
Thus, by Lemma~\ref{lem:partition-invariant}, $a' \in (\oddone)_H$ and $q' \in (\oddone \cup \unone)_H$, thus the edge $(a', q')$ continues
to get deleted in Step~4 of Algorithm~\ref{algo:pop-matching} when run on $H$.
This completes the proof of the lemma.
\end{proof}
Using the above lemmas we prove the following theorem.
\begin{theorem}
\label{thm:sonly-no-f-post}
Let $a_1 \in \A_s$. Then by falsifying her preference list alone,
she cannot get matched to a post $q \in f_G(a_1)$ in any popular matching in the falsified instance.
\end{theorem}
\begin{proof}
For contradiction assume that there exists a falsified instance $H$ such that
in a popular matching $M'$ of $H$, agent $a_1$ gets matched to $q \in f_G(a_1)$. By
Lemma~\ref{lem:sonly-not-in-cycle}, the post $q$ belongs to a non-sink component of $G_M$.
Further by Lemma~\ref{lem:tight-set-H}, there exists a set of agents $\A_q$ and a set of posts $\p_q$ such that
$|\A_q| = |\p_q|$, $a_1 \notin \A_q$
and for every $a \in \A_q$, we have $choices_H(a) \subseteq \p_q$.
Thus, if $a_1$ gets matched to $q$ in $M'$, then there is at least one agent $a' \in \A_q$
which does not have a post to be matched in $choices_H(a')$. This contradicts the fact that
$M'$ is a popular matching in $H$.
\end{proof}

\subsection{The modified instance $\modG$}
\label{sec:modified-instance}
As mentioned earlier, we need to define a modified instance, call it $\modG$ to
develop our cheating strategies. Recall from Example~\ref{ex:need-modified}
that 
a rank-1 edge which gets deleted from the graph $G'$ in Algorithm~\ref{algo:pop-matching},
can be used in a popular matching in a falsified instance. Thus, we define $\modG$
from the instance $G$
which has the following properties:
(i) every popular matching in $G$ is a popular matching in $\modG$ and,
(ii) any edge $(a,p)$ that gets deleted in Step~4 of Algorithm~\ref{algo:pop-matching}
when run on $\modG$ also gets deleted in Step~4 when Algorithm~\ref{algo:pop-matching}
is run on $H$ such that $H \succ G$ w.r.t. $a_1$. However, the definition of $\modG$ 
is independent of the agent $a_1$.

The graph $\modG$ is defined as follows: Let $G_1$ be the graph on rank-1 edges of $G$ and let
$\{q_1, \ldots, q_k \}$ be the set of {\em unreachable} posts in $G_1$.
Let us add to the instance $G$, a dummy agent $b$ whose preference
list is of length 1 and has all the {\em unreachable} posts in $G_1$ tied
as her rank-1 posts. That is, the preference list of $b$ can be written as $(q_1, \ldots, q_k)$.
The set of posts as well as the preference lists of all the agents $a \in \A$ remain the same as in $G$.
Formally, $\modG = (\modA \cup \p, \modE)$ where
$\modA = \A  \cup \{b\}$ and $\modE = E \cup \{(b, q_1), \ldots, (b, q_k)\}$ and each $(b, q_i)$ is
a rank-1 edge. By the choice of preference list of $b$, we note that $f_{\modG}(b) = \{q_1, \ldots, q_k\}$
and $s_{\modG}(b) = \ell(b)$, the unique last-resort post that we add for convenience.

We note that even after the addition of agent $b$, a maximum matching $M_1$ in $G_1$ continues
to be maximum matching in $\modG_1$. However, with respect to the partition
of vertices on rank-1 edges in $\modG$, every vertex is either {\em odd } or {\em even}
in $\modG_1$. We  show that addition of $b$ leaves the set  $s(a)$  unchanged for every agent $a \in \A$.

\begin{lemma}
For every $a \in \A$, we have $s_{\modG}(a) = s_G(a)$.
\end{lemma}
\begin{proof}
It suffices to show that $(\evenone)_{\modG} \cap \p = (\evenone)_G  \cap \p$. Let
$M_1$ be a maximum matching in $G_1$. Since $M_1$ is also a maximum matching in $\modG_1$,
partition the vertices of $\modA \cup \p$ w.r.t. $M_1$ in $\modG_1$. It is easy to see that the addition of agent $b$
only makes every post that was {\em unreachable} in $G_1$ as {\em odd} in $\modG_1$. Thus the set of {\em even} posts in
$G_1$ and $\modG_1$ is same which completes the proof.
\end{proof}

Now let $M$ be a popular matching in $G$, then let
$\tilde{M}$ denote the corresponding matching in $\modG$ such that
for every $a \in \A$ we have $\tilde{M}(a) = M(a)$ and $\tilde{M}(b) = \ell(b)$, the unique last-resort post
of $b$.
Note that $\tilde{M}$ is a maximum matching on rank-1 edges in $\modG$ and
for every $a \in \A$, we have $\tilde{M}(a) \in \{f_{\modG}(a) \cup s_{\modG} (a)\}$.
Finally $\tilde{M}(b) \in s_{\modG}(b)$ since $s_{\modG}(b) = \{\ell(b)\}$.
It is clear that $\tilde{M}$ satisfied both the properties of Theorem~\ref{thm:pop-mat} and therefore is
a popular matching in $\modG$.
We can now construct the switching graph $\modG_{\tilde{M}}$ w.r.t. $\tilde{M}$ in $\modG$.
Having made these definitions, we can now prove the following lemmas.

\begin{lemma}
\label{lem:edge-deletion}
Let $(a,p)$ be an edge which gets deleted in Step~4 of Algorithm~\ref{algo:pop-matching} run on $\modG$.
Then $(a,p)$ gets deleted in Step~4 when Algorithm~\ref{algo:pop-matching} is run on
any instance $H$ such that $H \succ G$ w.r.t. $a_1$.
\end{lemma}

\begin{proof}
As mentioned earlier all vertices in $\modG_1$ are either {\em odd} or {\em even}, hence if an
edge $(a,p)$ got deleted in Step~4 of Algorithm~\ref{algo:pop-matching}, then it implies that
$\{a, p\} \in (\oddone)_{\modG}$.
To prove the lemma statement,
it suffices to show that an agent or a post that was {\it odd} in $\modG_1$ does not become
{\it even} in $H_1$. Recall that a maximum matching $M_1$ in $G_1$ which leaves $a_1$  unmatched is also
a maximum matching in $\modG_1$ as well as in $H_1$.
Let $a \in \A$ be such that $a \in (\oddone)_{\modG}$. This implies that there exists an odd length alternating path $T_1$
in $\modG_1$ beginning at an unmatched post $p$. The path $T_1$ cannot contain $a_1$ since $a_1$ is unmatched in $M_1$.
For contradiction,
assume that $a \in (\evenone)_H$. Then there exists an even length alternating path $T_2$
in $H_1$ starting at an unmatched agent. This path has to begin at $a_1$, otherwise it was
already present in $\modG_1$ contradicting the fact that $a \in (\oddone)_{\modG}$.
However, if the path begins at $a_1$, since $a_1$ in unmatched in $M_1$, we
get an augmenting path by joining $T_1$ and $T_2$ which contradicts the maximality of $M_1$ in $H_1$.

Now consider $p \in \p$ such that $p \in (\oddone)_{\modG}$ and let $T_1$ denote the odd length alternating
path starting at an unmatched agent. The path $T_1$ has to begin at $a_1$, if not, the same path is present in
$H_1$ and hence we are done. Now  assume for the sake of contradiction that
$p \in (\evenone)_H$. Let $T_2$ be the even length alternating path with respect to $M_1$ in $H_1$ from an unmatched post
$p'$. The path $T_2$ cannot contain $a_1$ since $a_1$ is unmatched and hence joining the two paths $T_1$ and $T_2$
we get an augmenting path in $\modG_1$ contradicting the maximality of $M_1$ in $\modG_1$.
\end{proof}
\begin{lemma}
\label{lem:choices-modG-H}
Let $a \in \A \setminus \{a_1\}$ such that $\tilde{M}(a)$ belongs to a non-sink component of $\modG_{\tilde{M}}$. Let $H$ be an instance such that
$H \succ G$ w.r.t. $a_1$. Then $choices_{H}(a) \subseteq choices_{\modG}(a)$.
\end{lemma}

\begin{proof}
Recall that for any $a \in \A \setminus \{a_1\}$, $f_{\modG}(a) = f_G(a) = f_H(a)$ and
$s_{\modG}(a) = s_G(a) = s_H(a)$. We also know that, in any instance, for an agent $a$, $choices(a) \subseteq \{f(a) \cup s(a) \}$.
Thus, if for an agent $a \in \A \setminus \{a_1\}$, it were the case that $choices_{\modG}(a) = \{f_{\modG}(a) \cup s_{\modG}(a)\}$,
then the lemma statement holds trivially.
However due to deletion of edges in Step~4 and Step~9 of Algorithm~\ref{algo:pop-matching} when run on $\modG$, it
may be the case that $choices_{\modG}(a) \subset \{f_{\modG}(a) \cup s_{\modG}(a)\}$.
We note that since $M(a)$ belongs to a non-sink component it implies that both $\{a, M(a)\} \in (\untwo)_{\modG}$.
Therefore, by Claim~\ref{claim:no-del-Step9}, no edge incident on $a$ gets
deleted in Step~9. Further by Lemma~\ref{lem:edge-deletion}, it is clear that if an edge $(a,p)$
gets deleted in Step~4 of Algorithm~\ref{algo:pop-matching}
run in $\modG$, then the same edge gets deleted in Step~4 when run on $H$. This gives us the desired result that
$choices_{H}(a) \subseteq choices_{\modG}(a)$ and completes the proof of the lemma.
\end{proof}

\section{Single manipulative agent}
\label{sec:single-agent}
In this section we develop an efficient characterization of the conditions under which $a_1$ can falsify her
preference list.
We formulate the strategy of $a_1$ depending on whether $a_1 \in \A_s$ or $a_1 \in \A_{f/s}$.
Throughout, we assume that the true preference list of $a_1$ is denoted by $\mathcal{L} = P_1, \ldots, P_t, \ldots, P_l$
where $P_i$ denotes the set of $i$-th ranked posts of $a_1$. Further, $f_G(a_1) = P_1$ and $s_G(a_1) \subseteq P_t$.
We will use the modified instance $\modG$ to formulate our strategies.  


\subsection{ $\A_s$ agent}
Let $a_1 \in \A_s$  and let $M$ be any popular matching in $G$
and $\tilde{M}$ denote the corresponding popular matching in $\modG$ which matches $b$ to $\ell(b)$. It follows
from the definition of $\A_s$ that, $M(a_1) = \tilde{M}(a_1) \in s_G(a_1)$ and therefore $M(a_1) \in P_t$.
We first characterize whether $a_1$ can get {\em better always} using 
the graph $\modG$ and the switching graph $\modG_{\tilde{M}}$.

Our cheating strategy for $a_1$ (as shown in Figure~\ref{algo:better-always-sonly}) 
is simple: it checks if any of $a_1$'s $i$-th ranked posts $p \in P_i$
where $i = 2 \ldots t-1$, either belongs to a sink component in $\modG_{\tilde{M}}$ or
has a path to $\tilde{M}(a_1)$ in $\modG_{\tilde{M}}$. If there exists such a post $p$, our strategy
ensures that every popular matching in the falsified instance $H$ matches $a_1$ to $p$.
We denote by $\mathcal{L}_f$ the falsified preference list of $a_1$.
We now state the main theorem in this section.

\begin{figure}[h]
\fbox{
\begin{minipage}[t]{0.95\textwidth}
\begin{enumerate}
\item For $i = 2 \ldots t-1$ check if there exists a post $p \in P_i$ in $a_1$'s preference list such that
\begin{enumerate}
\item $p$ belongs to a sink component in $\modG_{\tilde{M}}$ or,
\item $p$ has a path  to $\tilde{M}(a_1)$ in $\modG_{\tilde{M}}$.
\end{enumerate}
\item If no post satisfies (a) or (b) above, then true preference list $\mathcal{L}$ is optimal for $a_1$.
\item Else let $p$ denote the most preferred post of $a_1$ satisfying one of the above two properties. Set
post $p$ as $a_1$'s rank-1 post in the falsified preference list.
\item To obtain the rank-2 post for $a_1$, let $a_2$ be some agent such that $\tilde{M}(a_2) \in f_G(a_1)$. Let $p' \in s_G(a_2)$.
Set $p'$ as the rank-2 post of $a_1$ in the falsified instance.
\item  $\mathcal{L}_f = p, p'$.
\end{enumerate}
\end{minipage}
}
\caption{Cheating strategy for $a_1 \in \A_s$. }
\label{algo:better-always-sonly}
\end{figure}
\begin{theorem}
\label{thm:main1}
Let $a \in \A_s$. Then there exists a cheating strategy for $a_1$ to get {\em better always} 
if and only if there exists a post $p$ ranked $2 \ldots t-1$ on $a_1$'s preference list
satisfying either

(a) $p$ belongs to a sink component in $\modG_{\tilde{M}}$ or,

(b) $p$ has a path to $\tilde{M}(a_1)$ in $\modG_{\tilde{M}}$.
\end{theorem}
\begin{proof}
We break down the proof into necessity and sufficiency.

{\it Necessity:} Assume that a post $p$ satisfying one of the two properties of Theorem~\ref{thm:main1} exists.
Let $\mathcal{L}_f = p, p'$ be the falsified preference list for $a_1$
as returned by Step~5 of Figure~\ref{algo:better-always-sonly}. Let
$H$ denote the instance where $a_1$ submits $\mathcal{L}_f$ and the
rest of the agents are truthful.
We begin by noting that $s_H(a_1) = p'$. This is because $p' \in s_G(a_2)$ and hence $p' \in (\evenone)_G$, therefore
$p' \in (\evenone)_H$.
We show the following hold:
\begin{itemize}
\item {\bf There exists a popular matching $M'$ in $H$ such that $M'(a_1) = p$ :}
Let $M_1 = \tilde{M} \setminus \{(a_1, \tilde{M}(a_1)), (b, \tilde{M}(b))\}$, thus the post $\tilde{M}(a_1)$ is unmatched
in $M_1$.
If
$p$ had a path to $\tilde{M}(a_1)$, then let $T$ denote a path from $p$ to $\tilde{M}(a_1)$ in the graph $\modG_{\tilde{M}}$.
Else if $p$ belongs to a sink component $\X$ of $\modG_{\tilde{M}}$, then
let $T$ denote a path from $p$ to a sink in $\X$.
We note that the path $T$ does not contain $\tilde{M}(b)$ since $\tilde{M}(b) = \ell(b)$ which
does not have any incoming edge.
In either case, apply the path $T$ to $M_1$ to get another matching $M_2$, that is, $M_2 = M_1 \cdot T$.
Finally $M' = M_2 \cup (a_1, p)$.
Since $p$ is a post ranked $2, \ldots, t-1$ in $a_1$ true preference list,
and the rank of posts in $s_G(a_1)$ is exactly $t$,
it implies that $p \in (\oddone \cup \unone)_G$.
Therefore we can conclude that that $p \in (\oddone)_{\modG}$.
Further, in either case when the path $T$ ends in a sink vertex or it ends in $\tilde{M}(a_1)$, the end point of the path is
a post which belongs to $(\evenone)_{\modG}$. Thus the path $T$ has weight $w(T) = -1$ (refer Table~\ref{tab:edge-weights}).

We prove that $M'$ is popular $H$. Note that for every
$a \in \A \setminus \{a_1\}$, we have $M'(a) \in \{f_{\modG}(a) \cup s_{\modG}(a)\}$ which implies that
$M'(a) \in \{f_H(a) \cup s_H(a)\}$.  Also, $M'(a_1) = p$ and note that $f_H(a_1) = \{p\}$.
Finally, it remains to show that $M'$ is a maximum matching on rank-1 edges of $H$. To
see this, note that, $w(T) = -1$, therefore the number of rank-1 edges in $M_2$ is exactly one less than
the number of rank-1 edges in $\tilde{M}$. However, since $a_1$ gets matched to her rank-1 post
in $M'$, the number of rank-1 edges in $\tilde{M}$ and $M'$ are exactly the same, thus proving that
$M'$ is in fact a maximum matching in $H_1$. Thus, $M'$ is popular in $H$.

\item {\bf Every popular  matching in $H$ matches $a_1$ to $p$ :} For the purpose
of proving this part we will work with the graph $G$ and the switching graph $G_M$ corresponding
to a popular matching $M$ in $G$.
For contradiction assume that there
exists a popular matching $M''$ in $H$ such that $M''(a_1) = p'$.
Note that
the rank-2 post $p'$ of $a_1$ is chosen as follows. Let $a_2$ be an agent such that
$\tilde{M}(a_2) \in f_G(a_1)$. Then $p' \in s_G(a_2)$.
Recall that $\tilde{M}$ is a popular matching in $\modG$ which is obtained from a popular matching $M$ in $G$.
Let $M(a_2) = \tilde{M}(a_2) = q$. Note that since $q \in f_G(a_1)$,
by Lemma~\ref{lem:sonly-not-in-cycle}, $q$ belongs to a non-sink component in $G_M$. Further, by Lemma~\ref{lem:tight-set-H},
there exists a set $\A_{q}$ and $\p_{q}$ such that $|\A_{q}| = |\p_{q}|$, $a_1 \notin \A_{q}$
and for every $a \in \A_{q}$, we have $choices_H(a) \subseteq \A_{q}$.
We show that the post $p'$ also belongs to $\p_q$ and therefore if $M''(a_1)$ matches $a_1$
to $p'$ then there exists at least one agent $a \in \A_q$ who does  not have a post to be matched in $choices_H(a)$.
Thus, $M''$ cannot be a popular matching in $H$.
It remains to prove that $p' \in \p_q$.
Note that $q = M(a_2) \in f_G(a_2)$ and $p' \in s_G(a_2)$. If $G_M$ contains the edge $(q,p')$,
then we are done since by definition of $\p_q$, it is clear that $p' \in \p_q$. The edge $(q,p')$ can be absent in $G_M$
only if the edge $(a_2, p')$ gets deleted in Step~9 of Algorithm~\ref{algo:pop-matching}. Note that
$(a_2, p')$ cannot get deleted in Step~4 since only rank-1 edges get deleted in Step~4.
Since $q$ belongs to
a non-sink component in $G_M$, it implies that $q \in (\untwo)_G$. Therefore $M(q) =  a_2$ also belongs to $(\untwo)_G$.
Thus by Claim~\ref{claim:no-del-Step9},
the edge $(a_2, p')$ does not get deleted in Step~9 of Algorithm~\ref{algo:pop-matching} and hence $p' \in \p_{q}$.
This completes the proof.
\end{itemize}

{\it Sufficiency:}
To prove that our strategy in Figure~\ref{algo:better-always-sonly}
is optimal, we note that when $a_1$ was truthful, she got matched to her $t$-th ranked
post in every popular matching in $G$. Using our strategy either she gets matched to her $k$-th ranked post
in every popular matching in $H$ where $k < t$ or we declare that true preference list is optimal in which case
she remains matched to her $t$-th ranked post where $k = t$.
Then there exists no instance $H$ such that a popular matching in $H$ matches $a_1$ to a post $q'$ which
$a_1$ strictly prefers to her $k$-th ranked post.

For contradiction assume that there exists such an instance $H$ obtained by falsifying the preference list
of $a_1$ alone. Let $M'$ be  some popular matching of $H$ such that $M'(a_1) = q$ and $a_1$
strictly prefers $q$ to her $k$-th ranked post.
Since our strategy in Figure~\ref{algo:better-always-sonly} did not find $q$, the post
$q$ belongs to a non-sink component $\Y$ in  $\modG_{\tilde{M}}$ and further there exists
no path from $q$ to $\tilde{M}(a_1)$ in $\modG_{\tilde{M}}$.
Now consider the two sets $\p_{q}$ and $\A_{q}$  as defined by Lemma~\ref{lem:tight-sets}.
We know  that $|\A_{q}| = |\p_{q}|$.
Further, for every $a \in \A_{q}$, we have $\{choices_{\modG}(a)\} \subseteq \p_{q}$.
Note that, $a_1 \notin \A_{q'}$, otherwise $\tilde{M}(a_1) \in \p_{q'}$ which implies that there exists
a  path from $q$ to $\tilde{M}(a_1)$, a contradiction.
Further, note that $\ell(b) \notin \p_q$ and therefore, $b \notin \A_q$.
Thus, for every $a \in \A_{q}$ we have
$\{choices_H(a)\} \subseteq \p_{q'}$.  Therefore, if $M'$ matches $a_1$ to $q'$, there exists at least one agent $a \in \A_{q'}$ who
does not have a post to be matched in $choices_H(a)$ and hence
$M'$ is not a popular matching in $H$.

This finishes the proof of the theorem.

\end{proof}

\subsection { $\A_{f/s}$ agent}
Let $a_1 \in \A_{f/s}$ when she submits her true preference list. In order to get {\em better always}, the goal of $a_1$ is
to falsify her preference list such that every popular matching in the falsified instance $H$ matches
$a_1$ to posts in $P_1$.

Let $M$ be a popular matching in $G$ such that $M(a_1) = p$ and $p \in f_G(a_1)$. 
Let $\tilde{M}$ denote the corresponding popular matching in $\modG$ which matches $b$ to $\ell(b)$.
Consider the switching graph $\modG_{\tilde{M}}$.
Our strategy for $a_1$ to get better always (as described in Figure~\ref{algo:better-always-fands}) is to
search for an {\em even} post $p'$ in $G_1$ which belongs to a non-sink component of $\modG_{\tilde{M}}$ and further the post $p'$
does not have a path $T$ to $\tilde{M}(a_1)$ in $\modG_{\tilde{M}}$ where $w(T) = +1$.
We prove the correctness and optimality of our strategy using the following theorem.
\begin{figure}[h]
\fbox{
\begin{minipage}[t]{0.95\textwidth}
\begin{enumerate}
\item For every $p' \in (\evenone)_G \cap \p$ check if
\begin{enumerate}
\item $p'$ belongs to a non-sink component, say $\Y_1$,  of $\modG_{\tilde{M}}$ and,
\item $p'$ does not have a path $T$ to $\tilde{M}(a_1)$ in $\modG_{\tilde{M}}$ such that $w(T) = +1$.
\end{enumerate}
\item If no post satisfies both properties, declare true preference list $\mathcal{L}$ is optimal for $a_1$.
\item Else set $M(a_1) = p$ and $p'$ as the rank-1 and rank-2 posts respectively in the falsified preference list of $a_1$.
\item  $\mathcal{L}_f = p, p'$.
\end{enumerate}
\end{minipage}
}
\caption{Cheating strategy for $a_1 \in \A_{f/s}$ to get better always.}
\label{algo:better-always-fands}
\end{figure}

\begin{theorem}
\label{thm:main2}
Let $a_1 \in \A_{f/s}$. There exists a cheating strategy for $a_1$ to get {\em better always}
if and only if there exists a post $p'$  in $(\evenone)_{G}$ satisfying the following two properties

(a) $p'$ belongs to a non-sink component, say $\Y_1$, of $\modG_{\tilde{M}}$, and

(b) there exists no path $T$ from $p'$ to $\tilde{M}(a_1)$ in $\modG_{\tilde{M}}$ such that $w(T) = +1$.
\end{theorem}
\begin{proof}
We break down the proof into necessity and sufficiency.

{\it Necessity:}
Assume that there exists a post $p'$ satisfying both the properties of Theorem~\ref{thm:main2} exists
and let $H$
be the instance obtained when $a_1$ submits $\mathcal{L}_f = p, p'$ as output
by Step~4 of Figure~\ref{algo:better-always-fands}. We show that every popular matching in $H$ matches $a_1$ to $p$.
It is easy to see that $M = \tilde{M} \setminus \{(b, \tilde{M}(b))\}$ is  popular in $H$.
Recall that, by choice, $M$ is in fact a popular matching in $G$ such that $M(a) = p$ and $p \in f_G(a_1)$.
Since $M$ is a maximum matching on rank-1 edges of $G$,
it continues to be a maximum matching on rank-1 edges of $H$. Further, for every $a \in \A \setminus \{a_1\}$, we have
$M(a) \in \{f_G(a) \cup s_G(a)\}$ which implies that $M(a) \in \{f_H(a) \cup s_H(a)\}$. Thus $M$ is a popular matching in $H$.

We now show that every popular matching of $H$ matches $a_1$ to $p$. Assume not.
Then let $M'$ be a popular matching in  $H$ such that  $M'(a_1) = p'$.
By the choice of $p'$, the post $p'$ belongs to a non-sink component say $\Y_1$
of $\modG_{\tilde{M}}$ and $p'$ does not have a  path to $\tilde{M}(a_1) = p$. Now let us define \tightpair
$\p_{p'}$ and $\A_{p'}$ as in Lemma~\ref{lem:tight-sets}. Thus,
we have $|\p_{p'}| = |\A_{p'}|$ and every $a \in \A_{p'}$ satisfies $choices_{\modG}(a) \subseteq \p_{p'}$.
Since $p'$ does not have a path to $\tilde{M}(a_1)$ in $\modG_{\tilde{M}}$, it is clear that $\tilde{M}(a_1) \notin \p_p'$. Therefore
$a_1 \notin \A_{p'}$. Further, $\tilde{M}(b) \notin \p_p'$ since $\tilde{M}(b) = \ell(b)$ does not have any incoming edges.
Thus, for every $a \in \A_{p'}$, we conclude that  $choices_H(a)  \subseteq \p_{p'}$.
Now if $M'(a_1) = p'$ then there exists at least one
agent $a \in \A_{p'}$ which does not  have a post to be matched in $choices_H(a)$.
Thus, $M'$ cannot be a popular matching in $H$.

{\it Sufficiency:}
To prove the other direction, assume that no post satisfying the two properties of Theorem~\ref{thm:main2}
exists. And for contradiction assume that there exists an instance $H$
where every popular matching in $H$ matches $a_1$ to a post in $f_G(a_1)$.
Let  $q' \in s_H(a_1)$. First note that $q' \notin f_G(a_1)$
because otherwise $q' \in (\oddone)_G$ and then $q' \notin (\evenone)_H \cap \p$.
Since our algorithm did not find
$q'$, it must be the case that $q'$ either belongs to a sink component of $\modG_{\tilde{M}}$
or if $q'$ belongs to a non-sink component of $\modG_{\tilde{M}}$, then $q'$ has a positive weight
path to $\tilde{M}(a_1)$.
In each case we construct a popular matching $M'$ in $H$ such that $M'(a_1) = q'$.
This gives us the desired contradiction.
Now note that since $\tilde{M}(a_1) \in f_G(a_1)$ and $f_G(a_1) \subseteq (\oddone)_G$, there is at least one agent, say $a_2$,
such that $\tilde{M}(a_2) = M(a_2) \in s_G(a_2)$.

\begin{itemize}
\item $q'$ belongs to a sink component $\X_i$ of $\modG_{\tilde{M}}$: Since $q'$ belongs to a sink component $\X_i$ there
exists a path $T$ starting at $q'$
which ends in a sink in $\X_i$. Now assume  that $\tilde{M}(a_1) \notin \X_i$ or if $\tilde{M}(a_1) \in \X_i$ then,
$q'$ does not have a path to $\tilde{M}(a_1)$.
In this case, consider the matching $M_1 = \tilde{M} \cdot T$ in which
the post $q'$ is unmatched. Let $M_2 = M_1 \setminus \{(a_1, \tilde{M}(a_1))\}$. Finally $M' = \{(a_1, q'), (a_2, \tilde{M}(a_1))\} \cup  M_2$.

We first note that $w(T) = 0$. This is because $q' \in s_H(a_1)$ implies $q' \in (\evenone)_H \cap \p = (\evenone)_{\modG} \cap \p$
and the end point of the path $T$ is a sink vertex in $\modG_{\tilde{M}}$ which again belongs to $(\evenone)_{\modG} \cap \p$.
Thus from Table~\ref{tab:edge-weights}, it is clear that $w(T) = 0$.
Let $\A_T = \cup_{p \in T} \{\tilde{M}(p)\}$ denote the set of agents matched to posts in $T$.
Since $q'$ did not have a directed path to $\tilde{M}(a_1)$, the agent $a_1 \notin \A_T$.
Further, $b \notin \A_T$ since $\tilde{M}(b) = \ell(b)$ does not have an incoming edge.
Thus for every agent
$a \in \A_T$, we have $M'(a) \in \{f_{\modG}(a) \cup s_{\modG}(a)\}$ which implies that $M'(a) \in \{f_H(a) \cup s_H(a)\}$.
Further, $M'(a_2) \in f_{\modG}(a_2) = f_H(a_2)$ and $M'(a_1) \in s_H(a_1)$.
Therefore, for every $a \in \A$, $M'(a) \in \{ f_H(a) \cup s_H(a)\}$.
Finally, since $w(T) = 0$ and $a_1$ and $a_2$ compensate for the rank-1 edges amongst themselves,
it is clear that the number of rank-1 edges in $M'$ is the same as number of rank-1 edges in $\tilde{M}$.
Thus, $M'$ is
a maximum matching on rank-1 edges of $H$.
Thus $M'$ is a popular matching in $H$ such that $M'(a_1) = q'$ which gives us the desired contradiction.

The case when $q'$ has a path to $\tilde{M}(a_1)$ is handled below.

\item $q'$ belongs to a non-sink component $\Y_j$ of $\modG_M$:
Since $q'$ belongs to a non-sink component $\Y_j$,
it implies that $q'$ has a directed path $T$ to $\tilde{M}(a_1)$ such that $w(T) = +1$.
Otherwise our strategy in Figure~\ref{algo:better-always-fands}
would have found $q'$.
The case when $q'$ belongs to a sink component, let
$T$ denote the path from $q'$ to $\tilde{M}(a_1)$.
Since  $q' \in (\evenone)_{\modG}$
and $\tilde{M}(a_1) \in (\oddone)_{\modG}$, using Table~\ref{tab:edge-weights},
it is clear that $w(T) = +1$.

We obtain $M'$ as follows:
let $M_1 = \tilde{M} \setminus \{(a_1, \tilde{M}(a_1)), (b, \tilde{M}(b))\}$. This leaves the post $\tilde{M}(a_1)$ unmatched in $M_1$.
Let $M_2 = M_1 \cdot T$ and finally let $M' = M_2 \cup (a_1, q')$.
Using the same arguments as above
it is possible to show that for every $a \in \A$, $M'(a) \in \{f_H(a) \cup s_H(a)\}$.
We note that since $w(T) = +1$ and $a_1$ no longer remains matched to one of her rank-1 posts,
the number of rank-1 edges in $M'$ and $\tilde{M}$ is the same.
Thus, $M'$ is a maximum matching on rank-1 edges in $H$. Therefore, $M'$ is popular in $H$
and $M'(a_1) \in s_H(a_1)$ which gives us the required contradiction.
\end{itemize}

\REM{
We mention that although the strategy is decided with respect to a particular popular matching $M$ in $G$,
the checks done are independent of the popular matching.
This is because by Lemma~\ref{lem:bad1},  if there exists a path $T$ between two vertices $p$ and
$q$  in
a non-sink component of $G_M$, then for any popular matching $M'$ and the corresponding switching graph $G_{M'}$
we know that there exists a path $T'$ in $G_{M'}$ such that $w(T') = w(T)$.
This completes the proof of the theorem.
}

\end{proof}

Using Theorem~\ref{thm:main1} and Theorem~\ref{thm:main2} we conclude the following.

\begin{theorem}
\label{thm:final-result}
The optimal falsified preference list for a single manipulative agent to 
get {\em better always} can be computed in $O(\sqrt{n}{m})$ time if preference
lists contain ties and in time $O(m+n)$ time if preference lists are all strict.
\end{theorem}
\begin{proof}
The main steps of our strategy are (i) to compute the set of popular pairs,
(ii) to construct the switching graph, (iii) run the algorithm given by Figure~\ref{algo:better-always-sonly}
or Figure~\ref{algo:better-always-fands} as appropriate for the single manipulative agent.
We note that we use the modified graph $\modG$ for computing our strategies and 
let $\tilde{n}$ and $\tilde{m}$ denote the vertices and edges in $\modG$ respectively. Clearly,
$\tilde{n} = n+1$ and $\tilde{m}  < m+n = O(m)$.
Once the switching graph is constructed, we observe that the algorithms in Figure~\ref{algo:better-always-sonly}
and Figure~\ref{algo:better-always-fands} have checks which can be done in time which is linear in the size of
the switching graph. Thus the steps (i) and (ii) defined above decide the complexity of
our cheating strategy. In case of ties, we have shown that both the steps can be computed in $O(\sqrt{n}m)$
time. In case of strict lists, using the  switching graph given by McDermid and Irving~\cite{MI11},
both the steps can be computed in $O(m+n)$ time.
Thus we have the desired result.
\end{proof}

{\bf Remark:} In each case we constructed a falsified preference list for $a_1$
which was strict and of length exactly two. However, by appending the rest of the posts in $\p$
at the end of $a_1$'s preference list, there is no change in the popular matchings that the
instance $H$ admits. Thus, we conclude that, if an agent can manipulate to get {\em better always}
she can achieve the same when preference lists are required to be complete.

\section{A characterization of equilibrium}
\label{sec:equil}
Here we consider the set of agents, their preference lists and the popular
matchings algorithm as a complete information game. That is, knowing the true preference
lists of all the agents and that the central authority chooses an arbitrary
popular matching, every agent chooses a preference list for herself. This preference
list is then submitted to the central authority. The goal of every agent is to
get {\em better always}.
An equilibrium of the game is a set of preference lists, one for each agent,
such that no single agent can improve her situation by deviating from her
equilibrium preference list \cite{Osborne}.

We now show a necessary and sufficient condition for the true preference lists of the agents
to be an equilibrium of the above game. Let $G = (\A \cup \p, E)$ denote the instance where ranks on the
edges represent true preferences of the agents and let $M$ be a popular matching in $G$.
Let $A = \A_{f} \cup \A_{s} \cup \A_{f/s}$ denote the partition of agents with respect to 
the popular pairs in $G$.
Let $\modG = (\tilde{\A} \cup \p, \tilde{E})$ denote the graph as defined in Section~\ref{sec:modified-instance} and let $\tilde{M}$ denote the
corresponding popular matching which matches $b$ to $\ell(b)$. Note that the
agent $b$ is dummy and is not a part of the game.
Let $\modG_{\tilde{M}}$ denote the corresponding switching graph. 
Let $\modA = \modA_{f} \cup \modA_{s} \cup \modA_{f/s}$ denote the partition of agents with respect to 
the popular pairs in $\modG$.
We first state the following two lemmas.

\begin{lemma}
\label{lem:eq1}
If every connected component of $\tilde{G}_{\tilde{M}}$ is a sink component, then the set of true preference lists 
of agents $a \in \A$ is an equilibrium of the game defined by $G$.
\end{lemma}

\begin{proof}
We claim that since every connected component is a sink component, every $a \in \modA$ is such that $a \in {\modA}_f \cup {\modA}_{f/s}$.
Assume for contradiction that, there exists an agent $a \in {\modA}_s$. Then by Lemma~\ref{lem:sonly-not-in-cycle}, we know that
$\tilde{M}(a)$ belongs to a non-sink component of $\modG_{\tilde{M}}$. However, there are no non-sink components, therefore $a \notin \modA_s$
and hence $\modA_s = \phi$.
We claim that this also implies
that $\A_s = \phi$. Thus every $\A = \A_f \cup \A_{f/s}$.
If $a \in \A_f$, she does not deviate from her true preference list assuming
that the rest of the agents are truthful. On the other hand, if $a \in \A_{f/s}$, from Theorem~\ref{thm:main2}, $a$
can get {\em better always} if and only if there exists a post belonging to a non-sink component of $G_M$ with additional properties.
However, there is no non-sink component. Thus, the set of true preference lists is an equilibrium.

It remains to show that if $\modA_s = \phi$ then $\A_s = \phi$.
Assume not. Then there exists an $a' \in \A$ such that $a' \in \A_s$. However, $a' \in \modA_f \cup \modA_{f/s}$.
Therefore, there exists a popular matching $M'$ in $\modG$ such that $M'(a') \in f_{\modG}(a') = f_{G}(a')$. Now $M'$ has to
match $b$ to one of her rank-1 posts, otherwise if $M'(b) = \ell(b)$, then $M_1 = M' \setminus \{(b, \ell(b))\}$ is a popular
matching in $G$, a contradiction to the fact that $a' \in \A_s$. Now, recall by the definition of the dummy agent $b$,
the set of rank-1 posts for $b$ were the {\em unreachable} posts in the graph $G_1$. Therefore, there must exists
one agent say $a'' \in \A$ such that $M'(a'') \in s_{\modG}(a'') = s_{G}(a'')$. Consider the matching
$M_1 = M' \setminus \{(b, M'(b))\}$ which leaves $M'(b)$ unmatched. Finally, consider
$M_2 = M_1 \cup (a'', M'(b))$. It is easy to see that $M_2$ is a popular matching in $G$ such that $M_2(a') \in f_{G}(a')$,
a contradiction to the fact that $a' \in \A_s$.
\end{proof}

\begin{lemma}
\label{lem:eq2}
If there exists at least two non-sink components in $\tilde{G}_{\tilde{M}}$, then the set of true preference lists 
of agents $a \in \A$ is  not an equilibrium of the game defined by $G$.
\end{lemma}

\begin{proof}

Let $\Y_1$ and $\Y_2$ be two non-sink components in $\modG_{\tilde{M}}$. We first claim that any non-sink component
there exists an $a \in \A$ such that $a \in \A_{f/s}$. So let $a_1 \in \Y_1$ be such that $a_1 \in \A_{f/s}$.
Further, we claim that any non-sink component contains
a post $p \in \p$ such that $p \in (\evenone)_G$.
Let $p_2 \in \Y_2$ such that $p_2 \in (\evenone)_G$.
Assuming these two claims, it is easy to see that the set of true preference lists is not an equilibrium
because $a_1$ can falsify her preference list by choosing $p_2$ as her rank-2 post and
ensuring that every popular matching in the falsified instance matches $a_1$ to one of her true rank-1 posts.

It remains to show that the the two claims hold. First let us prove that, given a non-sink component in $\modG_{\tilde{M}}$,
there exists a $p \in \p$ such that $p \in (\evenone)_G$. Consider a non-sink component $\Y$ containing a post in $p \in \p$
such that $p \in (\oddone \cup \unone)_{\modG}$. Since $p$ belongs to a non-sink component, therefore,
$p \in (\untwo)_{\modG}$ and therefore $\tilde{M}(p) = a \in (\untwo)_{\modG}$. Therefore the edge $(a, q)$ where
$q \in s_{\modG}(a)$ does not get deleted and hence $q \in \Y$. However, note that since $q \in s_{\modG}(a)$,
therefore $q \in (\evenone)_{\modG} \cap \p = (\evenone)_{G} \cap \p$. Thus, the claim holds.

To show that in a non-sink component $\Y$ of $\modG_{\tilde{M}}$, there exists an agent $a \in \A_{f/s}$,
we show that there exists at least one $\evenone \oddone$ edge in $\Y$ which belongs to a cycle. This claim
can be  verified using the Table~\ref{tab:edge-weights}.

\end{proof}
Finally, we are left with the case when $\tilde{G}_{\tilde{M}}$ contains exactly one non-sink component and zero or more sink components.
In this case if there exists an agent $a \in \A_{f/s}$ such that $\tilde{M}(a)$ belongs to a sink component of $\tilde{G}_{\tilde{M}}$, we conclude
that the set of true preferences is not an equilibrium. This is because $a$ can falsify her preference list
by choosing her $s(a)$ contained in the non-sink component. In the remaining case, assume that there is no
agent $a \in \A_{f/s}$ such that $\tilde{M}(a)$ belongs to a sink component of $\tilde{G}_{\tilde{M}}$.
Then for every $a \in \A_s$ such that $\tilde{M}(a)$ belongs to the non-sink component in $\tilde{G}_{\tilde{M}}$, 
we can verify using our strategy for a single manipulative
agent whether true preference list is optimal for $a$. This can be done in time 
proportional to the size of the preference list of $a$.
Thus, given the switching graph, it is clear that in linear time we can verify whether the set of true preference lists
is an equilibrium strategy for the above defined game.
We therefore conclude the following:
\begin{theorem}
There exists an $O(\sqrt{n}m)$ time algorithm to decide whether true preference lists of the agents
are an equilibrium of the game defined above when preference lists contain ties and an $O(m+n)$ time
algorithm for the same when preference lists are all strict.
\end{theorem}

\subsection*{ Conclusion} In this paper we presented cheating strategies for a single manipulative agent to get {\em better always}.
We also studied the equilibrium of a non-cooperative game with all agents. It would be interesting to study how two
or more agents co-operate and falsify their preference lists in order to get {\em better always}. We leave this as an open problem.
Another contribution of the paper is the switching graph characterization of the popular matchings problem with ties. McDermid and Irving~\cite{MI11}
have used their characterization in case of strict lists to give efficient algorithms for the {\em optimal} popular matchings
problem ~\cite{MI11} ~\cite{KN08}. It would be useful to exploit the characterization developed here and design efficient algorithms for the {\em optimal}
popular matchings problem with ties allowed. We leave that as another open question.

\noindent {\bf Acknowledgment:} The author is grateful to Prof. Vijaya Ramachandran for useful discussions on the problem
and also to an anonymous reviewer for the useful comments which improved the presentation considerably.
\bibliographystyle{abbrv}
\bibliography{references}


\end{document}